\documentclass{article}

\usepackage{PRIMEarxiv}

\usepackage[utf8]{inputenc} 
\usepackage[T1]{fontenc}    
\usepackage{hyperref}       
\usepackage{url}            
\usepackage{booktabs}       
\usepackage{nicefrac}       
\usepackage{microtype}      
\usepackage{lipsum}
\usepackage{cite}
\usepackage{xcolor}
\usepackage{amsmath,amssymb,amsfonts}
\usepackage{fancyhdr}       
\usepackage{graphicx}       
\graphicspath{{media/}}     
\usepackage{algorithm}
\usepackage{algpseudocode}
\newtheorem{theorem}{Theorem}
\newtheorem{remark}{Remark}
\newtheorem{definition}{Definition}
\newtheorem{assumption}{Assumption}
\newtheorem{corollary}{Corollary}
\newtheorem{lemma}{Lemma}
\usepackage{makecell}
\newtheorem{proof}{Proof}
\usepackage{tikz}
\usetikzlibrary{calc,positioning,decorations.pathreplacing}
\title{The Role of Identification in Data-driven Policy Iteration: A System Theoretic Study
}

\author{
  Bowen Song \\
  Institute for Systems Theory and Automatic Control \\
  University of Stuttgart \\
  Stuttgart, Germany\\
  \texttt{bowen.song@ist.uni-stuttgart.de} \\
   \And
  Andrea Iannelli \\
  Institute for Systems Theory and Automatic Control \\
  University of Stuttgart \\
  Stuttgart, Germany\\
  \texttt{andrea.iannelli@ist.uni-stuttgart.de} \\
}

\begin{document}
\maketitle

\begin{abstract}
The goal of this article is to study fundamental mechanisms behind so-called indirect and direct data-driven control for unknown systems. Specifically, we consider policy iteration applied to the linear quadratic regulator problem. 
Two iterative procedures, where data collected from the system are repeatedly used to compute new estimates of the desired optimal controller, are considered. 
In indirect policy iteration, data are used to obtain an updated model estimate through a recursive identification scheme, which is used in a certainty-equivalent fashion to perform the classic policy iteration update. 
By casting the concurrent model identification and control design as a feedback interconnection between two algorithmic systems, we provide a closed-loop analysis that shows convergence and robustness properties for arbitrary levels of excitation in the data. 
In direct policy iteration, data are used to approximate the value function and design the associated controller without requiring the intermediate identification step. After proposing an extension to a recently proposed scheme that overcomes potential identifiability issues, we establish under which conditions this procedure is guaranteed to deliver the optimal controller. Based on these analyses we are able to compare the strengths and limitations of the two approaches, highlighting aspects such as the required samples, convergence properties, and excitation requirement. Simulations are also provided to illustrate the results.
\end{abstract}


\section{Introduction}
Data-driven control, which aims at providing control design approaches for all those cases where a complete mathematical description of the system is not available, is a very active research area \cite{HOU20133,khaki2023introduction,10156081,Doerfler_CSM_23_Pt2}. This has several reasons, including the difficulty in modeling complex systems dominating modern applications only relying on first-principles and the increasing availability of data. 
This area comprises a vast body of works, differing for problem settings, techniques, and control objectives. However, there are two categories that fundamentally distinguish almost any proposed method, namely indirect and direct approaches. 
In indirect data-driven control, data is collected and used to learn an estimated model (and possibly its uncertainty), which is then used in combination with classic model-based control methodologies. On the other hand, direct data-driven control uses data to design the controller directly, somehow avoiding the intermediate step of identifying a model. Traditionally, the former has close connections with system identification \cite{ljung1999system} and, when the controller changes during operation, with indirect adaptive control \cite{AC}. Examples of the latter are instead model-free reinforcement learning (RL) \cite{1111} and direct adaptive control \cite{Annaswamy_23_ARCRAS}. 
Comparisons between these approaches to data-driven control have recently been framed in the context of behavioral systems theory \cite{9705109} and learning theory \cite{pmlr-v99-tu19a}.
Identifying the strengths and weaknesses of these approaches is a topic of ongoing research in the community, and achieving a better understanding of the fundamental differences requires looking at this from a different angle. 

In this work, a system theoretic approach is employed to analyze the role of the identified model, which represents the conceptual distinction between direct and indirect methods, in establishing important properties such as number of required samples, convergence rates, and robustness to the quality of the data.
For this we use the observation that iterative algorithms in learning and optimization can be viewed as dynamical systems \cite{Bhaya2006_SIAM_book_AlgoDyn,dorfler2024systems}. This viewpoint has been used before for analysis of optimization algorithms using robust control tools \cite{Scherer_SJCO21-Synth,Lessard_CSM_22_Diss} and the interconnections between a plant and an algorithm \cite{9446633,ZANELLI2021109901}.

To be concrete, we select a prototypical problem setting that has received strong interest across the control and RL community, namely policy iteration (PI) for solving the linear quadratic regulator (LQR) problem.
Policy iteration is a dynamic programming algorithm for optimal control \cite{lewis2012optimal,Bertsekas2022_Abstr_DP} and 
is still 
extensively studied \cite{park2020structured,9794431,9444823,BERTSEKAS_Newton_method_22} as it is the backbone of many approximate dynamic programming algorithms used in RL. PI involves two steps: policy evaluation and policy improvement. 
In policy evaluation, the cost of a given policy is computed, which is then leveraged in the second step to obtain a new (better performing) policy. In its classic formulation, both steps require knowledge of a model, 
and convergence can be ensured under assumptions on cost and dynamics using monotonicity and contractivity properties of the underlying operators\cite{doi:10.1080/00207179.2015.1079737}. 

As for LQR, this is a cornerstone optimal control problem which is a very common setting to study data-driven algorithms \cite{pmlr-v99-tu19a,9444823,Mohammadi_TAC_Conv_mF_LQR,9691800,articlesimulation,Ferizbegovic_LCSS_2020,10383604,9511623}. 
Relevant works for this study are 
\cite{9444823}, which investigates the impact of addictive uncertainties in model-based PI for continuous-time LQR problems, and 
\cite{9691800}, which proposes a direct data-driven policy iteration for the case when system dynamics are unknown.  
In \cite{10383604} a data-driven policy gradient method which combines recursive least squares and model-based policy gradient is proposed, and the convergence properties of this interconnection are analyzed using averaging theory. In \cite{9511623} a model-free $\lambda-$policy iteration method which integrates both value iteration and policy iteration techniques is proposed. This method relaxes the initial stabilizing policy gains requirements typically imposed in PI.

The paper analyzes the properties of indirect and direct policy iteration methods to LQR for unknown systems. In both cases we consider online schemes whereby controller updates take place concurrently with the generation of new data from the system.  
In the indirect policy iteration (IPI) case, we consider the interconnection of a recursive least squares (RLS) estimation algorithm with a policy iteration scheme. The gain matrix is iteratively updated through the PI steps applied to the model estimate provided by RLS from recursive least squares which uses online data measured from the system. 
We propose an algorithmic dynamical systems viewpoint to this problem and see the iterative scheme as a nonlinear feedback interconnection for which we analyze, notably without making a priori assumptions on the level of excitation of the data, the closed-loop response. 
We can use these results to establish under which conditions the system converges to the desired values (optimal controller and true system's model) and, if not, we can provide a guaranteed upper bound on their distance. We can capture analytically the lack of excitation in the data used by the algorithm as a source of disturbance and we are thus able to provide an input-to-state stability-type result that has an intuitive and useful in practice interpretation. 
In the direct policy iteration (DPI) case, the PI steps are parametrized by quantities that can be estimated directly via the data, without explicit model identification. For this we build on a recently developed on-policy (i.e. data are generated with the policy that is being updated) direct policy iteration method \cite{9691800} and modify it to address identifiability issues arising in noise-free scenarios. On the basis of this analysis, we thoroughly discuss: fundamental limitations of the schemes compared to the case where a model is available; strengths and weaknesses of the indirect and direct formulations; and point out clear advantages, at least for this setting, of the indirect solution. These advantages can be summarized as a better sample complexity (in the sense of number of required samples to provide estimates with a guaranteed accuracy) and more robustness to lack of excitation in the data. We reflect on the role played by the availability of an estimated model. 

The main contributions of this work can be summarized as follows:
\begin{enumerate}
    \item investigation of the properties of the recursive least squares algorithm without traditional persistence of excitation requirements,    
    \item system-theoretical analysis of the indirect data-driven policy iteration method,     
    \item extension of a recently proposed direct data-driven policy iteration method to address identifiability issues,    
    \item quantitative comparative study of the indirect and direct data-driven policy iteration methods.         
\end{enumerate}

The paper is organized as follows. Section \ref{sec:PSP} introduces the problem setting and provides some preliminaries. Section \ref{sec:IDD} and Section \ref{sec:DDD} detail the methodologies of the indirect policy iteration and direct policy iteration approaches and analyze the convergence properties, respectively. Section \ref{sec:discussion} focuses on comparing the indirect and direct policy iteration approaches, while Section \ref{sec:simulation} shows the effectiveness of two proposed policy iteration methods and verifies the comparisons by simulations. Section \ref{sec:conclusion} serves as a concluding summary of the work. 

\subsection*{Notation}
We denote by $A\succeq 0$ and $A\succ0$ positive semidefinite and positive definite matrix $A$, respectively. The symbol $\mathbb{S}^n_+$ denotes the sets of real $n\times n$ symmetric and positive semidefinite matrices. $\mathbb{Z}_+$ and $\mathbb{Z}_{++}$ are the sets of non-negative integers and positive integers. Kronecker product is represented as $\otimes$ , $vec(A)=[a_1^\top,a_2^\top,...,a_n^\top]^\top$ stacks the columns of matrix $A$ into a vector, $vecv(v)=[v_1^2,2v_1v_2,...,2v_1v_n,v_2^2,...,2v_2v_n,...,v_n^2]^\top$ rearranges the entries of vector $v$ in this specific pattern, $vecs(P)=[p_{11},...,p_{1n},p_{22},...,p_{2n},...,p_{nn}]^\top$ stacks the upper-triangular part of matrix $P\in \mathbb{S}^n_+$. For matrices, $\lVert \cdot\rVert_F$, $\lVert \cdot\rVert_2$ and $\lVert \cdot\rVert_\infty$ denote respectively their Frobenius norm, induced $2$-norm and infinity norm. $I_n$ and $O_n$ are the identity matrix and zero matrix with $n$ row/columns, respectively. The symbols $\lfloor x \rfloor$ and $\lceil x \rceil$ denote the floor function which returns the greatest integer smaller or equal than $x\in\mathbb{R}$ and ceil function which returns the smallest integer greater or equal than $x\in\mathbb{R}$, respectively.  A function belongs to class $\mathcal{K}$ if it is continuous, strictly increasing, and vanishing at the origin. A function $\beta(x , t)$ is called $\mathcal{KL}$ function if $\beta(x , t)$ decrease to 0 as $t\rightarrow0$ for every $x\geq0$ and $\beta(\cdot , t)\in \mathcal{K}$ for all $t\geq 0$. A sequence is a map $\mathbb{Z}_+\rightarrow \mathbb{R}^{n \times m}$ and is denoted by $Y:=\{Y_t\}$. 
For a scalar sequence ($n=m=1$), we define its infinite norm as $\lVert Y \rVert_\infty=\sup\limits_{t\in{\mathbb{Z}_+}}\{Y_t\}$. 
The symbols $\lambda(A)$ and $\lambda_i(A)$ denote the set of eigenvalues and the smallest $i$-th eigenvalue of the square matrix $A$, respectively.   

\section{Problem Setting and Preliminaries}\label{sec:PSP}
We consider discrete-time linear time-invariant systems of the form
\begin{equation}\label{LTI}
  x_{t+1}=Ax_t+Bu_t,
\end{equation}
where $x_t\in \mathbb{R}^{n_x}$ is the system state, $u_t\in \mathbb{R}^{n_u}$ is the control input and $t$ is the timestep. The matrices $A\in \mathbb{R}^{n_x \times n_x}$ and $B \in \mathbb{R}^{n_x \times n_u}$ are unknown, but we assume that the pair $(A,B)$ is stabilizable. This is a standard assumption in data-driven control \cite{pmlr-v99-tu19a},\cite{9691800}. The objective is to design a state-feedback controller $u_t=Kx_t$ that minimizes the following infinite horizon cost:
\begin{equation}\label{Cost}
  J(x_t,K)=\sum\limits_{k=t}^{+\infty} r(x_k,u_k)=\sum\limits_{k=t}^{+\infty} x_k^\top Qx_k+u_k^\top Ru_k,
\end{equation}
where $R\succ 0$ and $Q\succeq 0$. 
Because the dynamics (\ref{LTI}) is linear and the cost (\ref{Cost}) is quadratic, this optimal control problem is referred to as an infinite horizon linear quadratic regulator (LQR).
Given a linear state-feedback control that is stabilizing (i.e. $u_t=K x_t$ and the matrix $A+BK$ is Schur stable), the corresponding cost $J(x_t, K)$  
can be expressed as $x_t^\top Px_t$, where $P\succ 0$ is also called the quadratic kernel of the cost function associated with $K$ \cite{9691800}. Starting from \eqref{Cost}, $P$ can be calculated by solving the following dynamic programming equation named value-based Bellman equation: 
\begin{equation}\label{Bellman}
    \underbrace{x_t^\top Px_t}_{J(x_t,K)}=\underbrace{x_t^\top Qx_t+u_t^\top Ru_t}_{r(x_t,u_t)}+\underbrace{x_{t+1}^\top Px_{t+1}}_{J(x_{t+1},K)}.
\end{equation} 
which does not depend on the system model, but only on the state response. 
Using \eqref{LTI} in \eqref{Bellman}, the model-based Bellman equation reads:
\begin{equation}\label{MBBE}
  P=Q+K^\top RK+(A+BK)^\top P(A+BK).
\end{equation} 
It is a well-known result from optimal control \cite{FERRANTE2013471} that the optimal controller solution to the LQR problem is a linear state-feedback, and the optimal feedback gain $K^*$ is obtained via:
\begin{subequations}\label{Kpolicyimprovement_EQ}
\begin{align}
K^*&=-(R+B^\top P^*B)^{-1}B^\top P^*A, \label{Kpolicyimprovement} \\
P^*&=Q+A^\top P^*A-A^\top P^*B(R+B^\top P^*B)^{-1}B^\top P^*A. \label{DARE} 
\end{align}
\end{subequations}
Here, $P^*$ is the quadratic kernel of the value function, i.e. of the cost associated with the optimal gain $K^*$, and is the unique solution of the discrete algebraic Riccati equation (DARE) \eqref{DARE}.

\subsection{Model-based policy iteration}
Even when the system's model is available, solving \eqref{DARE} directly is challenging especially when dealing with a high number of system states. Policy iteration offers an effective iterative approach to finding the optimal gain $K^*$. The basic version of the policy iteration algorithm\cite{1099755}, which requires knowledge of the system's matrices $A$ and $B$, is summarized in Algorithm \ref{Algo1}.
\begin{algorithm}
  \caption{Model-based policy iteration.}\label{Algo1}
  \begin{algorithmic}
      \Require $A,B$, a stabilizing policy gain $K_1$ 
      \For{$i=1,...,+\infty$}
        \State \textbf{Policy Evaluation: find $P_{i}$} 
        \State $P_{i}=Q+K_i^\top RK_{i}+(A+BK_{i})^\top P_{i}(A+BK_{i})$
        \State \textbf{Policy Improvement: update gain $K_{i+1}$}
        \State $K_{i+1}=-(R+B^\top P_iB)^{-1}B^\top P_iA$
      \EndFor
  \end{algorithmic}
\end{algorithm}

Algorithm \ref{Algo1} consists of two phases. In the policy evaluation phase, a stabilizing policy gain $K_i$ is evaluated, i.e. its performance according to the cost function \eqref{Cost} is calculated by computing the corresponding quadratic kernel $P_i$ with the model-based Bellman equation \eqref{MBBE}. In the policy improvement phase, a new gain $K_{i+1}$ is obtained by minimizing the following cost:
\begin{equation}\label{costPI}
  K_{i+1}=\mathop{\mathrm{arg~min}}\limits_{K}\underbrace{\left[ x_t^\top Qx_t+(Kx_t)^\top R(Kx_t)+(Ax_t+BKx_t)^\top P_i(Ax_t+BKx_t)\right]}_{=:J_{\mathrm{PI}}(K)}.
\end{equation}
Equation \eqref{costPI} can be interpreted as treating $P_i$ as if it was the quadratic kernel of the value function, 
and then minimizing the associated infinite horizon cost $J_{\mathrm{PI}}(K)$. 
Indeed this optimization problem is unconstrained and strictly convex because $R+B^\top P_iB \succ 0$. Thus, to minimize ${J_{\mathrm{PI}}(K)}$ the first order derivative  of \eqref{costPI} with respect to $K$ should be $0$ for every $x_t$, yielding:
\begin{equation}\label{costPI1}
   \frac{\partial J_{\mathrm{PI}}(K)}{\partial K}={\left(2\left( R+B^\top P_iB\right) K+2B^\top P_iA\right)} \left( x_t x_t^\top \right),
\end{equation}
and leads indeed to:
\begin{equation}\label{costPI2}
\begin{split}
     2(R+&B^\top P_iB) K_{i+1}+2B^\top P_iA=0 \Rightarrow    K_{i+1}=-(R+B^\top P_iB)^{-1}B^\top P_iA.
\end{split}
\end{equation} 
The properties of Algorithm \ref{Algo1} are summarized below.
\begin{theorem}{Properties of model-based PI \ref{Algo1}\cite{1099755}} \label{theorem1}
 \\
If the system dynamics $(A,B)$ are controllable, then 
  \begin{enumerate}
    \item $P_1\succeq P_2 \succeq ... \succeq P^*$;
    \item $K_i$ stabilizes the system $(A,B)$, for all $i=1,2,...$;
    \item $\lim\limits_{i\rightarrow\infty}P_i=P^*$, $\lim\limits_{i\rightarrow\infty}K_i=K^*$;
    \item Convergence rate of $P_i$ is given by $\lVert P_{i+1}-P^*\rVert_F \leq C\lVert P_{i}-P^*\rVert^2_F$, where $C$ is a constant.
  \end{enumerate}
\end{theorem}
\begin{remark}
   The controllability requirement in Theorem \ref{theorem1} can be relaxed to stabilizability as per \cite[Theorem 3]{HEWER1973226}.
\end{remark}

\subsection{Online parameter estimation with RLS}
When the dynamics are unknown but linear, one can use linear regression 
to identify the system's model\cite{ljung1999system}. System \eqref{LTI} can be expressed as: 
\begin{equation}\label{ConvergenceSystemID1}
    x_{t+1}=Ax_t+Bu_t=\underbrace{\left[A~B\right]}_{=:\theta}\underbrace{\left[\begin{array}{cc}
         x_t  \\
         u_t 
    \end{array}\right]}_{=:d_t}.
\end{equation}
Given data $\{d_k,x_{k+1}\}^T_{k=1}$ from a trajectory of length $T$, a classic choice for estimating $\hat{\theta}$ is to minimize a least-squares loss function\cite{pmlr-v99-tu19a,articlesimulation,Ferizbegovic_LCSS_2020}:
\begin{equation}\label{Solution}
  \begin{split}
     \theta \in\mathop{\mathrm{arg~min}}\limits_{\hat{\theta}} \sum\limits_{k=1}^{T}(x_{k+1}-\hat{\theta} d_k)^\top(x_{k+1}-\hat{\theta} d_k).
  \end{split}
\end{equation}
When $H_T:=\left(\sum\limits_{k=1}^{T}d_kd_k^\top\right)$ is invertible, $\theta$ is estimated exactly via \eqref{Solution} as follows: 
\begin{equation}\label{Solution2}
  \begin{split}
     \theta=\left(\sum\limits_{k=1}^{T}x_{k+1}d_k^\top\right)H_T^{-1}
  \end{split}
\end{equation}
This is also called batch least squares because the unknown is obtained in one step using all the available data. In an online setting where data are not available all at once but are collected progressively, one can opt for
the recursive least squares (RLS) algorithm \cite{AC} which updates the estimate iteratively based on new measured information. Define $\hat{\theta}_0:=[\hat{A}_0,\hat{B}_0]$ the initialization of system dynamics and $\hat{\theta}_t:=[\hat{A}_t,\hat{B}_t]$ the estimates at time $t$ of the estimation process. 
In RLS the new estimate $\hat{\theta}_t$ based on the new trajectory $\{d_t,x_{t+1}\}$ and previous estimates $\hat{\theta}_{t-1}$ is obtained in analogy to (\ref{Solution2}) as:
\begin{equation}\label{COnvergenceSystemid8}
  \begin{split}
     \hat{\theta}_t &=\left(\sum\limits_{k=1}^{t}x_{k+1}d_k^\top\right)
     H_t^{-1}\\
       &=\left(\sum\limits_{k=1}^{t-1}x_{k+1}d_k^\top+x_{t+1}d_t\right)H_t^{-1} \\
       &=\left(\hat{\theta}_{t-1}H_{t-1}+x_{t+1}d_t\right)H_t^{-1} \\
       &=\left(\hat{\theta}_{t-1}(H_{t}-d_td_t^\top)+x_{t+1}d_t\right)H_t^{-1} \\
      &=\hat{\theta}_{t-1}+(x_{t+1}-\hat{\theta}_{t-1}d_{t})d^\top_{t}H_t^{-1}.
  \end{split}
\end{equation}
where $H_t$ has also a recursive update rule:
\begin{equation}\label{ConvergenceSystemid9}
  H_{t}=H_{t-1}+d_td_t^\top,
\end{equation}
and $H_0 \succ 0$ is an initialized matrix guaranteeing invertibility of $H_t$ at all $t$. As we will see shortly, the convergence of the RLS algorithm is governed by properties of the sequence $\{d_t d_t^\top\}$.

The RLS scheme is summarized in Algorithm \ref{Algo4}.
\begin{algorithm}
  \caption{Recursive least squares.}\label{Algo4}
  \begin{algorithmic}
      \Require an initialized estimated system dynamic $\hat{\theta}_0$ and positive definite matrix $H_0$ 
      \For{$t=1,...,\infty$}
          \State \textbf{Given $\{x_{t+1},d_t\}$}
          \State $H_{t}=H_{t-1}+d_td_t^\top$
          \State $\hat{\theta}_t=\hat{\theta}_{t-1}+(x_{t+1}-\hat{\theta}_{t-1}d_{t})d^\top_{t}H_t^{-1}$
      \EndFor
  \end{algorithmic}
\end{algorithm}
\begin{remark}\label{rank1update}
In practice, one never inverts $H_{t}$ in \eqref{COnvergenceSystemid8}. For example, by using Sherman-Morrison formula \cite{Petersen2008}, $H_{t}^{-1}$ can be recursively updated as:
   \begin{equation}\label{rank1}
     H_{t}^{-1}=(H_{t-1}+d_{t-1}d_{t-1}^\top)^{-1}=H_{t-1}^{-1}-\frac{H_{t-1}^{-1}d_{t-1}d^\top_{t-1}H_{t-1}^{-1}}{1+d^\top_{t-1}H_{t-1}^{-1}d_{t-1}},
   \end{equation}
   which only requires matrix multiplications between $H_{t-1}^{-1}$ and $d_{t-1}$. 
\end{remark}
The estimation error $\Delta\theta_t:=\hat{\theta}_t-\theta$ at time index $t$ is given as:
\begin{equation}\label{convergenceSystemid4}
\begin{aligned}
  \Delta\theta_t=\Delta\theta_{t-1}\left( I_{n_x+n_u}-d_{t}d^\top_{t}H_t^{-1}\right)=\Delta\theta_{t-1}H_{t-1}H_{t}^{-1}=\Delta\theta_0H_0H_{t}^{-1}.   
\end{aligned}
\end{equation}
From \eqref{convergenceSystemid4}, an upper bound of $\lVert \Delta\theta_t \rVert_F$ is given as follows:
\begin{equation}\label{convergence6}
    \lVert \Delta\theta_t \rVert_F\leq\lVert \Delta\theta_0H_0\rVert_F\lVert H_{t}^{-1} \rVert_F=:\Delta\theta_t^\mathrm{Upper},
\end{equation}
where $\{\Delta\theta^\mathrm{Upper}_t\}$ is a well-defined sequence because the inverse of $H_t$ always exists.

To characterize the convergence of the RLS estimation error, we introduce the following properties of sequences.
\begin{definition}{(Global persistency\cite{bruce2020convergence})}\label{def1}\\
    A sequence $\{Y_t\} \in \mathbb{S}^{n}_+$ is (globally) persistent if there exist $N\geq 1$ and $\alpha > 0$ such that, for all $j\geq 0$,
    \begin{equation}\label{ConvergenceSystemid7}
        \sum\limits_{i=0}^{N-1}Y_{i+j} \succeq \alpha I_n.
    \end{equation}
    The numbers $\alpha$ and $N$ are respectively the lower bound and persistency window of $\{Y_t\}$. \\
    If $Y_t$ is not symmetric, then we use the convention that $\{Y_t\} \in \mathbb{R}^{n\times m}$ is (globally) persistent if $\{Y_t Y_t^\top\}$ is (globally) persistent.
\end{definition}
\begin{definition}{(Local persistency)}\label{def2}\\
    A sequence $\{Y_t\} \in \mathbb{S}^{n}_+$ is locally persistent if there exist $N\geq 1, M\geq 1$ and $\alpha > 0$ such that , for all $j=Mk+1$ where $k\in\mathbb{Z}_+$,
    \begin{equation}\label{ConvergenceSystemid8}
        \sum\limits_{i=0}^{N-1}Y_{i+j} \succeq \alpha I_n.
    \end{equation}
    The numbers $\alpha$ and $N$ are respectively, the lower bound and persistency window of $\{Y_t\}$. $M$ is the persistency interval. As before, $\{Y_t\} \in \mathbb{R}^{n\times m}$ is locally persistent if $(Y_t Y_t^\top)$ is locally persistent.
\end{definition}

Whereas Definition \ref{def1} is classic in the adaptive control literature \cite{AC}, to the best of the authors' knowledge, Definition \ref{def2} has not been proposed so far and is a weaker requirement that will prove useful for the coming results. The difference between global and local persistency is that inequality \eqref{ConvergenceSystemid7} must hold for all $j$, whereas \eqref{ConvergenceSystemid8} must only hold on a subset of indices. If a sequence is globally persistent with window $N$, then it is also locally persistent with the same window $N$ and $M=1$. The reverse is not true. For example, the following periodic scalar sequence $\{1,0,0,1,1,0,0,1,...,\}$ is locally persistent with $N=2$ and $M=2$ but not globally persistent with $N=2$.

It turns out that local persistency is sufficient to show convergence of the RLS estimate.
\begin{theorem}\label{theorem2}
    If the sequence $\{d_t\}$ is locally persistent with any lower bound $\alpha>0$, any persistency window $N\in \mathbb{Z}_{++}$ and $M\in \mathbb{Z}_{++}$, then RLS satisfies $\lim\limits_{t\rightarrow+\infty} \Delta\theta_t^\mathrm{Upper}=0$ and thus $\lim\limits_{t\rightarrow+\infty} \hat{\theta}_t=\theta$. 
\end{theorem}
The proof of Theorem \ref{theorem2} can be found in Appendix \ref{App1}.

\section{Indirect Data-driven Policy Iteration}\label{sec:IDD}

In this section, we analyze the indirect data-driven policy iteration (IPI), which combines the model-based policy iteration introduced in Algorithm \ref{Algo1} with the RLS Algorithm \ref{Algo4}. The proposed algorithm is a possible solution to perform PI in the case of unknown system dynamics by performing concurrently policy optimization and online system identification. Even though IPI is a straightforward combination of two standard methods, our primary focus is analyzing its convergence properties and limitations from a system theoretic viewpoint.

\subsection{Algorithm definition}

First, we introduce the structure of the algorithm, for which it is beneficial to use two sets of indices: one referring to \emph{episodes} (denoted by $i$) and the other to the conventional timesteps $t$.   
An episode is characterized by a collection of data $d_t$ from a start timestep to an end timestep. In IPI, the $i$-th episode is associated with the finite length sequence $\{d_t\}$ from time index $(i-1)\tau_\mathrm{IPI}+1$ to time index $i\tau_\mathrm{IPI}$, where $\tau_\mathrm{IPI}$ denotes the length of an episode. This is graphically illustrated in Figure \ref{timeepisode}.\\
\begin{figure}[htb]
  \centering
    \begin{tikzpicture}[auto, node distance=2.2cm]
\draw[->,ultra thick] (-3.5,1)--(2.5,1) node[right]{$i$-episode};
\draw[->,ultra thick] (-6,-1)--(6,-1) node[right]{$t$-timestep};
\draw[black] (-2.5,1) circle (2pt) ;
\filldraw[black] (-1.5,1) circle (2pt) node[above]{$i-1$};
\filldraw[black] (-0.5,1) circle (2pt) node[above]{$i$};
\filldraw[black] (0.5,1) circle (2pt) node[above]{$i+1$};
\draw[black] (1.5,1) circle (2pt) ;

\draw[black] (-5,-1) circle (2pt) ;
\draw[black] (-4,-1) circle (2pt) ;
\filldraw[black] (-3,-1) circle (2pt) node[below]{$\begin{array}{cc}
     ~  \\
     (i-1)\tau_{\mathrm{IPI}} 
\end{array}$};
\draw[black] (-2,-1) circle (2pt) ;
\draw[black] (-1,-1) circle (2pt) ;
\filldraw[black] (0,-1) circle (2pt) node[below]{$\begin{array}{cc}
     ~  \\
     i\tau_{\mathrm{IPI}} 
\end{array}$};

\draw[black] (1,-1) circle (2pt) ;
\draw[black] (2,-1) circle (2pt) ;
\draw[black] (3,-1) circle (2pt) ;
\filldraw[black] (4,-1) circle (2pt) node[below]{$\begin{array}{cc}
     ~  \\
     (i+1)\tau_{\mathrm{IPI}} 
\end{array}$};
\draw[black] (5,-1) circle (2pt) ;

\draw[thick, dashed] (-1.5,1)--(-4,-1);
\draw[thick, dashed] (-0.5,1)--(-3,-1);
\draw[thick, dashed] (-0.5,1)--(-0,-1);
\draw[thick, dashed] (0.5,1)--(1,-1);
\draw[thick, dashed] (0.5,1)--(4,-1);
 \end{tikzpicture}
  \caption{Illustration of episode indices $i$ and time indices $t$.}\label{timeepisode}
\end{figure}
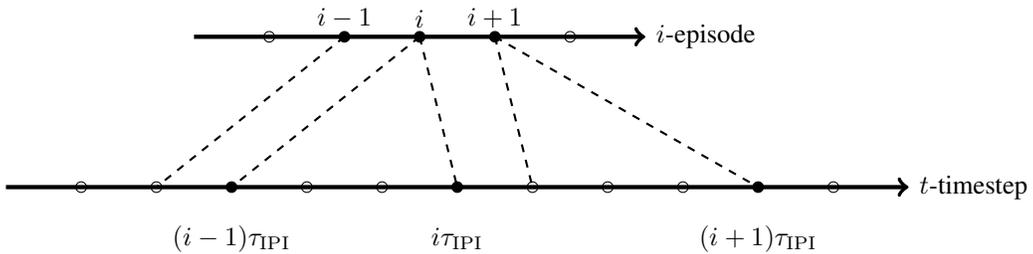

\begin{remark}
The reason to use this double index is that PI variables ($\hat{K}_{i}, \hat{P}_{i}$) and the RLS variables $(\hat{A}_i,\hat{B}_i)$ are updated after each episode $i$. However, the latter updates depend on data indexed by the timestep $t$ collected within each episode from time index $(i-1)\tau_\mathrm{IPI}+1$ to time index $i\tau_\mathrm{IPI}$. 
\end{remark}

At a generic episode $i\geq1$, the procedure of IPI involves the following steps: 

\begin{itemize}
  \item Given a stabilizing policy gain $\hat{K}_{i}$, which originates from the initialization $(i=0)$ or the previous episode, the evaluation step yielding $\hat{P}_i$ is done by solving the model-based Bellman equation \eqref{MBBE} using the available estimate $\left(\hat{A}_{i-1},\hat{B}_{i-1}\right)$:
      \begin{equation}\label{IPI2a}
        \hat{P}_{i}=Q+\hat{K}_i^\top R\hat{K}_{i}+\left(\hat{A}_{i-1}+\hat{B}_{i-1}\hat{K}_{i}\right)^\top\hat{P}_{i}\left(\hat{A}_{i-1}+\hat{B}_{i-1}\hat{K}_{i}\right).
      \end{equation}
      
  \item 
  The plant is excited with $u_t=\hat{K}_i x_t+e_t$, where $e_t$ is a potentially non-zero feedforward term, e.g. a dithering signal introduced for excitation purposes. The $\tau_\mathrm{IPI}$ data points collected from timestep $(i-1)\tau_\mathrm{IPI}+1$ to $i\tau_\mathrm{IPI}$ are then utilized to recursively update the estimates of $\left(\hat{A}_i,\hat{B}_i\right)$ through the RLS Algorithm:
      \begin{subequations}\label{procedure}
        \begin{align}
  H_{i} &=H_{i-1}+\left(\sum\limits_{t=(i-1)\tau_\mathrm{IPI}+1}^{i\tau_\mathrm{IPI}}d_td_t^\top\right),\label{IPI2b} \\
  \hat{\theta}_i &= \left(\hat{\theta}_{i-1}H_{i-1}+ \sum\limits_{t=(i-1)\tau_\mathrm{IPI}+1}^{i\tau_\mathrm{IPI}}x_{t+1}d_t^\top \right)H_{i} ^{-1}.\label{IPI2c}
         \end{align}
      \end{subequations}
  \item The policy is improved using the updated estimate $\left(\hat{A}_i,\hat{B}_i\right)$:
  \begin{equation}\label{IPI2d}
    \hat{K}_{i+1}=-\left(R+\hat{B}_{i}^\top\hat{P}_{i}\hat{B}_{i}\right)^{-1}\hat{B}_{i}^\top\hat{P}_{i}\hat{A}_{i}.
  \end{equation}
\end{itemize}

The choice of the excitation policy, which can be different than the one suggested here and does not even have to be \emph{on-policy}, as well as the impact of parameter $\tau_\mathrm{IPI}$, will be discussed in detail in Section \ref{selectiontau} and Section \ref{Selection} as part of the comparison with the direct counterpart. \\
\begin{remark}
    The term $e_t$ represents an additional degree of freedom of the online policy, which could be used for example as an exploratory signal normally or uniformly distributed, or made a function of $\hat{P}_{i}$ or $H_i$ for active learning purposes. We note here however that the subsequent analysis are agnostic to the choice of $e_t$ and remain valid also for the case $e_t=0$. 
\end{remark}

The IPI algorithm is graphically illustrated in Figure \ref{fig:stucture} and summarized in Algorithm \ref{Algo2}. The closed-loop system is formed by the physical system and the controller, connected by the solid black lines in Figure \ref{fig:stucture}. The algorithmic dynamics are formulated by PI and RLS, connected by the dashed black lines in Figure \ref{fig:stucture}. In the following section, our focus lies on the analysis of the algorithmic dynamics.


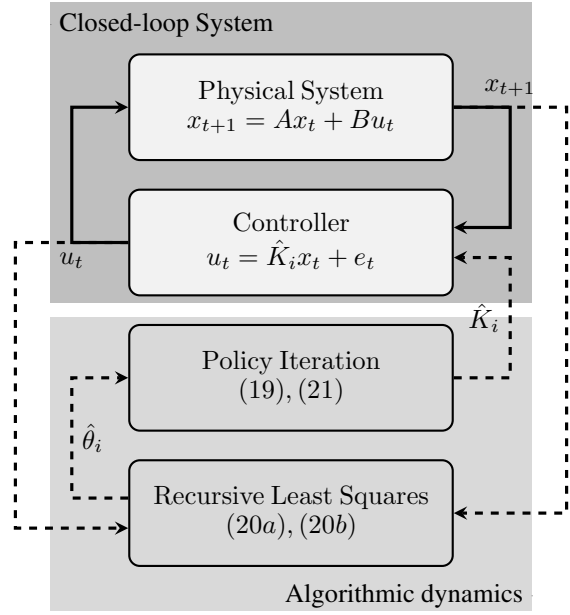
\begin{figure}[htb]
    \centering
    \begin{tikzpicture}[auto, node distance=1.8cm]
        \tikzstyle{block} = [rectangle, thick, draw=black!100, minimum width=4.3cm, minimum height=1.5cm,fill=gray!10, text centered, rounded corners, minimum height=4em]
        \tikzstyle{blockblue} = [rectangle, thick, draw=black!100, minimum width=4.3cm, minimum height=1.5cm, text centered, rounded corners, minimum height=4em]
        \tikzstyle{line} = [draw]
        \tikzstyle{arrowblack}=[->, >=stealth, very thick, black]
        \tikzstyle{arrowblue}=[->, dashed,>=stealth, very thick, black]     
        
        \fill[gray!50,semitransparent] (-3.2,-2.6) rectangle (3.2,1.4);
        \filldraw[black] (-3.2, 1.1) circle (0pt) node[below, right]{Closed-loop System};
        \fill[gray!30,nearly transparent] (-3.2,-6.7) rectangle (3.2,-2.8);
        \filldraw[black] (3.2, -6.5) circle (0pt) node[above,left]{Algorithmic dynamics};
        \node [block] (block1) {$\begin{array}{cc}\mathrm{Physical~System}\\ x_{t+1}=Ax_t+Bu_t \end{array}$};
        \node [block, below of=block1] (block2) {$\begin{array}{cc}\mathrm{Controller}\\ u_{t}=\hat{K}_ix_t+e_t \end{array}$};
        \node [blockblue, below of=block2] (block3) {$\begin{array}{cc}\mathrm{Policy~Iteration}\\ (19),(21) \end{array}$};
        \node [blockblue, below of=block3] (block4) {$\begin{array}{cc}\mathrm{Recursive~Least~Squares}\\ (20a),(20b) \end{array}$};
        \draw[arrowblack] (block1.east) -- ($(block1.east)+(0.75,0)$) -- ($(block2.east)+(0.75,0.2)$) -- ($(block2.east)+(0,0.2)$);
        \draw[arrowblack] (block2.west) -- ($(block2.west)+(-0.75,0)$) -- ($(block1.west)+(-0.75,0)$) -- (block1.west);
        \draw[arrowblue] (block1.east) --node[above,black]{$x_{t+1}$} ($(block1.east)+(1.5,0)$) -- ($(block4.east)+(1.5,0)$) -- (block4.east);
        \draw[arrowblue] ($(block4.west)+(0,0.2)$) -- ($(block4.west)+(-0.75,0.2)$) -- node[right,black]{$\hat{\theta}_i$}($(block3.west)+(-0.75,0)$) -- (block3.west);
        \draw[arrowblue] (block3.east) -- ($(block3.east)+(0.75,0)$) --node[left,black]{$\hat{K}_i$} ($(block2.east)+(0.75,-0.2)$) -- ($(block2.east)+(0,-0.2)$);
        \draw[arrowblue] (block2.west) --node[below,black]{$u_{t}$} ($(block2.west)+(-1.5,0)$) -- ($(block4.west)+(-1.5,-0.2)$) -- ($(block4.west)+(0,-0.2)$);

    \end{tikzpicture}
    \caption{Concurrent identification and policy iteration scheme.}
    \label{fig:stucture}
\end{figure}

\begin{algorithm}
  \caption{Indirect data-driven policy iteration.}\label{Algo2}
  \begin{algorithmic}
      \Require $\hat{A}_0,\hat{B}_0,H_0$, the initial stabilizing policy gain $\hat{K}_1$ 
      \For{$i=1,...,\infty$} (episode counter) 
        \State \textbf{Policy Evaluation: find $\hat{P}_{i}$} 
        \State $\hat{P}_{i}=Q+\hat{K}_i^\top R\hat{K}_{i}+\left(\hat{A}_{i-1}+\hat{B}_{i-1}\hat{K}_{i}\right)^\top\hat{P}_{i}\left(\hat{A}_{i-1}+\hat{B}_{i-1}\hat{K}_{i}\right)$
        \State \For{$j=1,...,\tau_\mathrm{IPI}$} (timestep counter inside episode $i$)
        \State \State \textbf{Excite the system with $u_{\tau_\mathrm{IPI}(i-1)+j}=\hat{K}_ix_{\tau_\mathrm{IPI}(i-1)+j}+e_{\tau_\mathrm{IPI}(i-1)+j}$}
        \State \State \textbf{Collect the data $\gets (x_{\tau_\mathrm{IPI}(i-1)+j},u_{\tau_\mathrm{IPI}(i-1)+j},x_{\tau_\mathrm{IPI}(i-1)+j+1})$}
        \State \State \textbf{Use RLS in Algorithm \ref{Algo4} to update $A_{temp},B_{temp},H_{temp}\gets A_{temp},B_{temp},H_{temp}, \mathrm{\textbf{data}}$}  
        \State \EndFor
        \State $\hat{A}_i,\hat{B}_i,H_i \gets A_{temp},B_{temp},H_{temp}$
        \State \textbf{Policy Improvement: update gain $\hat{K}_{i+1}$}
        \State $\hat{K}_{i+1}=-\left(R+\hat{B}_{i}^\top\hat{P}_{i}\hat{B}_{i}\right)^{-1}\hat{B}_{i}^\top\hat{P}_{i}\hat{A}_{i}$
      \EndFor
  \end{algorithmic}
\end{algorithm}


\subsection{Closed-loop analysis}\label{IPIConvergenceAnalysis}

As suggested by the representation in Figure \ref{fig:stucture}, we study policy iteration and recursive least-squares as a feedback interconnection of two coupled dynamical systems. In the "system PI", the inputs are the estimates $\left(\hat{A}_i,\hat{B}_i\right)$ obtained from the RLS, and the dynamics are described by \eqref{IPI2d} and \eqref{IPI2a}. 
In the "system RLS", the inputs are the data $\{d_t\}$ and $\{x_{t+1}\}$ collected online in an episode, and the dynamics are described by \eqref{IPI2b} and \eqref{IPI2c}.

We first analyze separately the two building blocks RLS (Section \ref{RLSsection}) and PI (Section \ref{PIsection}), and then provide the main results in Section \ref{IPIPI}. 
Before proceeding with these analyses, we introduce in Section \ref{PGS} new properties of (data) sequences that are convenient for the analysis of online algorithms when we cannot assume a priori that the properties of global or local persistence hold.

\subsubsection{Properties of general sequences}\label{PGS}
For the sake of this abstract discussion on sequences properties, we refer in the reminder to a generic sequence $\{S_i\}\in \mathbb{S}^{n}_+$. 
\begin{definition}{(Minimum persistency window length $N^\mathrm{PW}_i$ and lower bound $\alpha_i^{\mathrm{PW}}$)}\label{MinimumWindow}\\
  Given a sequence $\{S_i\}\in \mathbb{S}^{n}_+$, we define the minimum persistency window length $N^\mathrm{PW}_i \in \mathbb{Z}_{+}$ and its associated lower bound $\alpha_i^{\mathrm{PW}}\in\mathbb{R}_{+}$ at index $i$ as the two variables uniquely defined by these two conditions:
  \begin{itemize}
    \item the matrix $\sum\limits_{k=i}^{i+N^\mathrm{PW}_i-1}S_k$ is full rank, i.e. $\lambda_1\left(\sum\limits_{k=i}^{i+N^\mathrm{PW}_i-1}S_k\right)=:\alpha_i^{\mathrm{PW}}>0$,
    \item the matrix $\sum\limits_{k=i}^{i+N^\mathrm{PW}_i-2}S_k$ is singular, i.e. $\lambda_1\left(\sum\limits_{k=i}^{i+N^\mathrm{PW}_i-2}S_k\right)=0$.
  \end{itemize}
  If the sequence is not locally persistent, there are indices $i$ for which $N^\mathrm{PW}_i$ and $\alpha_i^{\mathrm{PW}}$ are not defined. In such cases, our convention is to set $N^\mathrm{PW}_i=0$ and $\alpha_i^{\mathrm{PW}}=+\infty$.
\end{definition}
Building on this notion, we define the largest persistency window length $\bar{N}$ and the smallest persistency window bound $\underline{\alpha}$ for a generic sequence $\{S_i\}$.
\begin{definition}{(Largest persistency window length $\bar{N}$ and smallest lower bound $\underline{\alpha}$)}\label{MinimumWindow2} \\
  Given a sequence $\{S_i\}\in \mathbb{S}^{n}_+$, we construct two sequences $N^{PW}=\{N^{PW}_i\}$ and $\alpha^{PW}=\{\alpha^{PW}_i\}$, whose entries are the minimum persistency window length $N^\mathrm{PW}_i$ and the persistency window lower bound $\alpha_i^{\mathrm{PW}}$ of the sequence $\{S_i\}$. The largest persistency window length and the smallest lower bound are then defined as $\bar{N}:=\sup\{N^\mathrm{PW}_i\}$ and $\underline{\alpha}:=\inf\{\alpha_i^{\mathrm{PW}}\}$
\end{definition} 

We then introduce an auxiliary locally persistent sequence $\{\tilde{S}_i\}$ associated with a generic sequence $\{S_i\}$ and 
defined by 
$N=N_{\max}$, $M=N_{\max}$ and $\alpha=\alpha_{\min}$ (recall Definition \ref{def2}), where
$N_{\max} \in [\bar{N},\infty)$ and $N_{\max} \neq 0$, $\alpha_{\min}\in (0,\underline{\alpha}]$ and $\alpha_{\min} \neq +\infty$. This sequence is shown as the thin line in Figure \ref{nonpersistent}. 
We define next the sequence of non-persistent eigenvalues associated with $\{S_i\}$, which is a useful measure of the difference between $\{S_i\}$ and its associated auxiliary locally persistent sequence $\{\tilde{S}_i\}$.

\begin{definition}{(Number of non-persistent eigenvalues $\{j^\mathrm{non}_i\}$ of a sequence $\{S_i\}$)}\label{def4}\\
   Given a sequence $\{S_i\}\in \mathbb{S}^{n}_+$ and the associated auxiliary locally persistent sequence with its entries $\{\tilde{S}_i\}$ defined by $N=N_{\max}$, $M=N_{\max}$ and $\alpha=\alpha_{\min}$. We can define the sequence of number
of non-persistent eigenvalues $\{j^\mathrm{non}_i\}$ with entries defined by:   
   \begin{equation}\label{definationeq1}
     \begin{split}
        j^\mathrm{non}_i=\max\limits_{j\in \{1,...,n\}}~j&  \\
        s.t.~ \lfloor\frac{i}{N_{\max}}\rfloor\alpha_{\min}-&\lambda_{j}\left(\sum\limits_{k=1}^{i}S_{k}\right) > 0. 
     \end{split} 
   \end{equation}
    Each entry $j_i^\mathrm{non}$ can be interpreted as the number of eigenvalues of $\sum\limits_{k=1}^{i}S_{k}$ which are upper-bounded by $\lfloor\frac{i}{N_{\max}}\rfloor\alpha_{\min}$, hence the name.    
When problem \eqref{definationeq1} is infeasible, our convention is to set $j^\mathrm{non}_i=0$. 
\end{definition}  
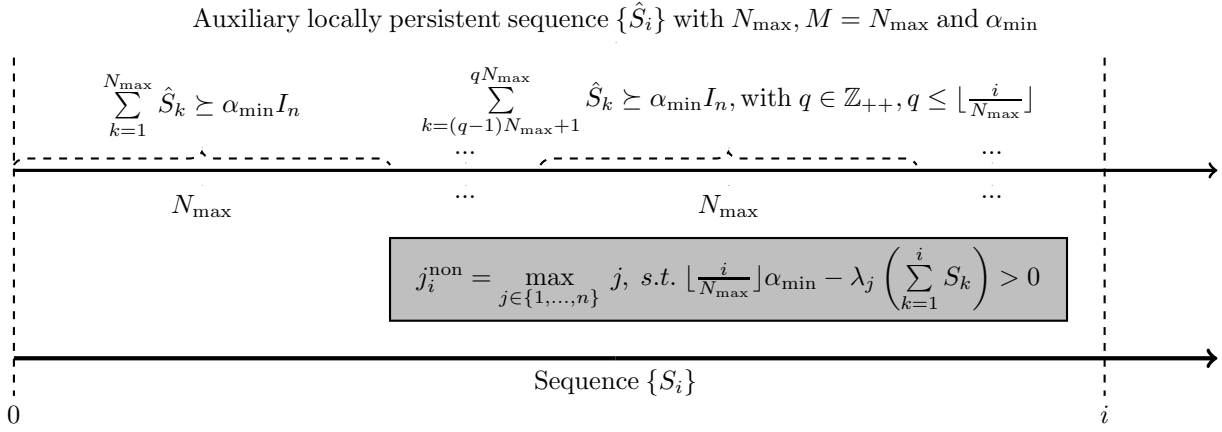
\begin{figure}
  \centering
      \begin{tikzpicture}[auto, node distance=2.2cm]
\draw[->,very thick, black] (-8,1)--(8,1);
\draw[ultra thick] (-8,-1.5)--(0,-1.5);
\draw[->,ultra thick] (0,-1.5) node[below]{$\mathrm{Sequence}~\{S_i\} $}--(8,-1.5);
\draw[thick,dashed] (-8,2.5)--(-8,-2) node[below]{$0$};
\draw[thick,dashed] (6.5,2.5)--(6.5,-2) node[below]{$i$};
\filldraw[black] (0,2.7) circle (0pt) node[above,black]{$\mathrm{Auxiliary~locally~persistent~sequence~} \{\hat{S}_i\} \mathrm{~with~} N_{\max},M=N_{\max} \mathrm{~and~} \alpha_{\min} $};

\draw[decorate, dashed, thick, decoration={brace, raise=2pt,aspect=0.5,amplitude=5pt}] (-8,1)--(-3,1);
\filldraw[black] (-5.5,1.3) circle (0pt) node[above,black]{${\sum\limits_{k=1}^{N_{\max}}\hat{S}_{k}}\succeq \alpha_{\min}I_n$};
\filldraw[black] (-5.5,0.8) circle (0pt) node[below,black]{${N_{\max}}$};


\filldraw[black] (-2,1.1) circle (0pt) node[above,black]{$...$};
\filldraw[black] (-2,0.8) circle (0pt) node[below,black]{$...$};

\draw[decorate, dashed, thick, decoration={brace, raise=2pt,aspect=0.5,amplitude=5pt}] (-1,1)--(4,1);
\filldraw[black] (1.5,1.3) circle (0pt) node[above,black]{${\sum\limits_{k=\left(q-1\right)N_{\max}+1}^{q N_{\max}}\hat{S}_{k}}\succeq \alpha_{\min}I_n, \mathrm{with}~q\in \mathbb{Z}_{++},q \leq \lfloor \frac{i}{N_{\max}}\rfloor$};
\filldraw[black] (1.5,0.8) circle (0pt) node[below,black]{${N_{\max}}$};

\filldraw[black] (5,1.1) circle (0pt) node[above,black]{$...$};
\filldraw[black] (5,0.8) circle (0pt) node[below,black]{$...$};

\fill[gray!50,nearly transparent] (-3,-1) rectangle (6,0.1);
\filldraw[black] (1.5,-1) circle (0pt) node[above]{$j^\mathrm{non}_i=\max\limits_{j\in \{1,...,n\}}~j, ~s.t.~ \lfloor\frac{i}{N_{\max}}\rfloor\alpha_{\min}-\lambda_{j}\left(\sum\limits_{k=1}^{i}S_{k}\right) > 0$};
\draw[thick] (-3,-1) rectangle (6,0.1); 

 \end{tikzpicture}
  \caption{Graphical representation of sequence properties.}\label{nonpersistent}
\end{figure}

\begin{remark}{(Interpretation of $j^\mathrm{non}$ and Figure \ref{nonpersistent})}\\
  We give here some insights on the meaning of the sequence $j^\mathrm{non}$ from Definition \ref{def4} and its pictorial representation in Figure \ref{nonpersistent}. The main intuition is that $\{j^\mathrm{non}_i\}$ is a measure of the difference between the generic sequence $\{S_i\}$ under investigation and its associated auxiliary locally persistent sequence $\{\tilde{S}_i\}$ introduced after Definition \ref{MinimumWindow}. Note first that, from its definition, $j^\mathrm{non}_i\in [0,n]$. Consider then the two \emph{extreme} cases 
    \begin{enumerate}
        \item $\{S_i\}$ is locally persistent with $N=N_{\max}$, $M=N_{\max}$ and $\alpha=\alpha_{\min}$. Then, it follows from Definition \ref{def2} that $\lambda_{1}\left(\sum\limits_{k=1}^{i}S_{k}\right) \geq\lfloor\frac{i}{N_{\max}}\rfloor\alpha_{\min}$, $\forall i\in \mathbb{Z}_+$, and thus problem \eqref{definationeq1} has no solution. This shows , for locally persistent sequences, $j^\mathrm{non}_i=0$, $\forall i\in \mathbb{Z}_{+}$.
                \item $\{S_i\}$ is a sequence with entries $S_i=0_n$, $\forall i\in \mathbb{Z}_+$.
        Then, $\bar{N}=0,\underline{\alpha}=+\infty$ (see Definition \ref{MinimumWindow} and \ref{MinimumWindow2}).   
        It follows from \eqref{definationeq1} that for any selection of $N_{\max}$ and $\alpha_{\min}$ it holds $j^\mathrm{non}_i=n$, $\forall~i\geq N_{\max}$. 
     \end{enumerate}
       For more general cases, observe that the auxiliary locally persistent sequence $\{\tilde{S}_{i}\}$ is characterized by the property that all $n$ eigenvalues of $\left(\sum\limits_{k=1}^{i}\tilde{S}_{k}\right)$ increase by at least $\alpha_{\min}$ after each persistent window. The latter, in view of the definition of local persistency, comprises $N_{\max}$ data points (indices from $i=mN_{\max}$ to $i=(m+1)N_{\max}$ with $m\in\mathbb{Z}_+$). This is illustrated by the thin line in Figure \ref{nonpersistent}. The smallest eigenvalue of $\left(\sum\limits_{k=1}^{i}\tilde{S}_{k}\right)$ can thus be lower-bounded by $\lfloor\frac{i}{N_{\max}}\rfloor\alpha_{\min}$. 
    If we then look carefully at \eqref{definationeq1}, which defines $j^\mathrm{non}_i$, we can interpret this integer variable as the number of eigenvalues of $\left(\sum\limits_{k=1}^{i}S_{k}\right)$ that are smaller than $\left(\lfloor\frac{i}{N_{\max}}\rfloor\alpha_{\min}\right)$. This explains why we can use this notion to quantify the difference between the sequence $\{S_i\}$ and its associated auxiliary locally persistent sequence $\{\tilde{S}_i\}$.\\
    A couple of other observations are in order. If a sequence is non-persistent, i.e. there exists $i_\mathrm{non}\in \mathbb{Z}_+$ such that at least one eigenvalue of $\left(\sum\limits_{k=1}^{i}S_{k}\right)$ does not increase for $i>i_\mathrm{non}$, then $j^\mathrm{non}_i\neq0$ $\forall~i \geq i_\mathrm{non}+N_{\max}$. Another interesting property of $j^\mathrm{non}_i$ is that $j^\mathrm{non}_i \leq j^\mathrm{non}_{i+1}$ $\forall i$, because the number of non-increasing eigenvalues is necessarily non-decreasing.
\end{remark}

\subsubsection{RLS analysis with a generic data sequence in the episodic setting}\label{RLSsection}
Standard convergence analysis of RLS assumes that the data sequence is globally \cite{bruce2020convergence} or, as in Theorem \ref{theorem2}, locally persistent. In this section, we study the properties of RLS applied to a generic data sequence by leveraging the definitions introduced in Section \ref{PGS}.  
To analyze RLS across episode $i$, it is convenient to define the sequence $\{D_i\}$ with entries $D_i$ given by:
\begin{equation}\label{DefinationDi}
  D_i:=\sum\limits_{j=(i-1)\tau_\mathrm{IPI}+1}^{i\tau_\mathrm{IPI}}d_jd_j^\top.
\end{equation} 
The main result can then be stated as follows.
\begin{theorem}\label{theorem3}
Given the sequence $\{D_i\}\in \mathbb{S}^{(n_x+n_u)}_+$ with entries defined in \eqref{DefinationDi}, the estimation error of the recursive least squares initialized with $\hat{\theta}_0$ and $H_0=aI, a>0$ is bounded by:
  \begin{equation}\label{RLS3}
  \begin{split}
      \lVert \hat{\theta}_i-\theta\rVert_F\leq \Delta\theta^\mathrm{Upper}_i\leq f\left(\lVert \hat{\theta}_0-\theta \rVert_F,i\right)+g\left(\lVert j^\mathrm{non}\rVert_{\infty}\right),
  \end{split}
\end{equation}
where: $f\left(\lVert \hat{\theta}_0-\theta \rVert_F,i\right):= a(n_x+n_u)\frac{\lVert \hat{\theta}_0-\theta\rVert_F}{a+\lfloor \frac{i}{N_{\max}}\rfloor{\alpha_{\min}}}$ and $g\left(\lVert j^\mathrm{non}\rVert_{\infty}\right):=\lVert \Delta\theta_0\rVert_F\lVert j^\mathrm{non}\rVert_{\infty}$; $\alpha_{\min}$ and $N_{\max}$ are parameters associated with the auxiliary locally persistent sequence of $\{D_i\}$; and $\{j^\mathrm{non}_i\}$ is the associated sequence (see Definition \ref{def4}). 
\end{theorem}
The proof of Theorem \ref{theorem3} can be found in Appendix \ref{ISSRLS}. 
This result recovers Theorem \ref{theorem2} when its conditions hold.
\begin{corollary}\label{Coro2}
  If the data sequence $\{d_t\}$ is locally persistent with any persistency window $N_{d_t}$, persistency interval $M_{d_t}$ and lower bound $\alpha_{d_t}$, for any positive integer $\tau_\mathrm{IPI} \in \mathbb{Z}_{++}$, then $\mathop{lim}\limits_{i\rightarrow\infty}\lVert \hat{\theta}_i-\theta\rVert_F=0$ and thus $\mathop{lim}\limits_{i\rightarrow\infty}\Delta\theta^\mathrm{Upper}_i=0$. 
\end{corollary}
\begin{proof}
      From Theorem \ref{theorem3}, we know that if the sequence $\{d_t\}$ is locally persistent as defined in Definition \ref{def2}, then the sequence $\{D_i\}$ is also locally persistent with persistency window length $M_{d_t}\lceil\frac{N_{d_t}}{M_{d_t}}\rceil$, persistency window interval $M_{d_t}\lceil\frac{N_{d_t}}{M_{d_t}}\rceil$and $\alpha_{d_t}\tau_\mathrm{IPI}$. As a result, $j_i^\mathrm{non}=0$ and thus the estimation error asymptotically converges to 0.
\end{proof}
When $\{d_t\}$ is not locally persistent, then Theorem \ref{theorem3} provides an upper bound for the estimation error which can be computed based on the properties of $\{D_i\}$ and its auxialiary sequence.

We finally remark that Theorem \ref{theorem2} stated the convergence of RLS with respect to the time index $t$, while Corollary \ref{Coro2} emphasizes the convergence of RLS with respect to the episode index $i$, which is relevant to the IPI setting. There is no conceptual difference between the two, as clear from the proof of Corollary \ref{Coro2}. In particular, it is important to observe that convergence is guaranteed for any selection of $\tau_\mathrm{IPI}$.  

\subsubsection{PI analysis with nominal system}\label{PIsection}

In a similar decoupling spirit, and as an introduction to the main result covered in the next section, we study here the convergence of policy iteration applied to the true system $(A,B)$. 

Let $\alpha(P_i):=B^\top P_iA$ and $\beta(P_i):=R+B^\top P_iB$ (where $\beta(P_i)$ is a positive definite matrix and thus is always invertible). Inserting the equation of policy improvement step \eqref{MBBE} into the policy evaluation step \eqref{costPI2}, the equation between $P_i$ and $P_{i+1}$ formulated by $A,B,Q,R$ reads:
\begin{equation}\label{relationPseq}
\begin{aligned}
   P_{i+1}&=Q+A^\top P_{i+1}A+\alpha^\top (P_i)\beta^{-1}(P_i)\beta(P_{i+1})\beta^{-1}(P_i)\alpha(P_i) \\
     &-\alpha^\top (P_{i+1})\beta^{-1}(P_i)\alpha(P_i)-\alpha^\top (P_{i})\beta^{-1}(P_i)\alpha(P_{i+1}).
\end{aligned}
\end{equation}
Using $ vec(EFG)=(F^\top  \otimes E) vec(G) $ from \cite{Petersen2008}, we can rewrite \eqref{relationPseq} as:
\begin{equation}\label{rerelationPseq2}
\begin{aligned}
      \underbrace{\left[I_{n_x}\otimes I_{n_x}-\left(A^\top - \alpha^\top (P_i)\beta^{-1}(P_i)B^\top \right)\otimes \left(A^\top -\alpha^\top (P_i)\beta^{-1}(P_i)B^\top \right)\right]}_{=:\mathcal{A}(P_i)}vec(P_{i+1})\\
      =vec\left(Q+\alpha^\top (P_i)\beta^{-1}(P_i)R\beta^{-1}(P_i)\alpha(P_i)\right).
  \end{aligned}
\end{equation}
If $\mathcal{A}(P_i)$ is invertible, we can define an operator $\mathcal{L}^{-1}_{(A,B,P_i)}$ that connects $P_i$ and $P_{i+1}$ with system matrix $(A,B)$ from \eqref{rerelationPseq2}:
\begin{equation}\label{ISS3}
   P_{i+1}=\mathcal{L}^{-1}_{(A,B,P_i)}\left(Q+\alpha^\top\left(P_i\right)\beta^{-1}\left(P_i\right)R\beta^{-1}\left(P_i\right)\alpha\left(P_i\right)\right).
\end{equation}
By connecting Theorem $\ref{theorem1}$ and \eqref{rerelationPseq2}, we develop the following theorem:
\begin{theorem}\label{theorem4}
    Given a stabilizable system $(A,B)$ and $P_0 \succeq P^*$ as the initial condition, for any $P_i$ which is recursively computed from $P_{i-1}$ using \eqref{rerelationPseq2} with $i\in\mathbb{Z}_{++}$, $\mathcal{A}(P_i)$ is always invertible and there always exists a constant $c\in (0,1)$, such that:
    \begin{equation}
        \left \lVert P_{i+1}-P^* \right \rVert_F \leq c\left \lVert P_{i}-P^* \right \rVert_F.
    \end{equation}
\end{theorem}
The proof of Theorem \ref{theorem4} is given in Appendix \ref{Prooftheorem6}. Whereas convergence of PI in this setting is a well-known result, Theorem \ref{theorem4} shows in addition that $P_i$ converges exponentially to $P^*$ when initialized as in the statement. A similar result was recently shown in \cite{9794431}, but we use here a different norm and proof technique, so this result can be considered novel and of independent interest.

\subsubsection{Convergence analysis of IPI}\label{IPIPI}

We are now at the stage of analyzing the convergence properties of the interconnection in Figure \ref{fig:stucture2} between the RLS and PI systems, studied independently in the previous Sections. 
\begin{figure}[htb]
 \centering
    \begin{tikzpicture}[auto, node distance=2.2cm]
        \tikzstyle{blockblue} = [rectangle, thick, draw=black!100, minimum width=4.5cm, minimum height=2cm,fill=gray!10, text centered, rounded corners, minimum height=4em]
        \tikzstyle{line} = [draw]

        \tikzstyle{arrowblue}=[->, >=stealth, very thick, black]
        \node [blockblue] (block3) {Policy Iteration};
        \node [blockblue, below of=block3] (block4) {Recursive Least Squares};
        \draw[arrowblue] (block3.east) --node[above,black]{$\hat{K}_i$} ($(block3.east)+(0.75,0)$);
        \draw[arrowblue] ($(block4.west)+(0,0)$) -- ($(block4.west)+(-0.75,0)$) -- node[right,black]{$\hat{\theta}_i$}($(block3.west)+(-0.75,0)$) -- (block3.west);
        \draw[arrowblue] ($(block4.east)+(0.75,-0.0)$) -- node[below,black]{$d_t$}  ($(block4.east)+(0,-0.0)$);
        \draw[thick, dashed] (3,1)--(3,-3.1);        
        \filldraw[black] (3, -1.1) circle (0pt) node[right]{$\begin{array}{cc}\mathrm{Closed-loop}\\ \mathrm{System} \end{array}$};
    \end{tikzpicture}
    \caption{Convergence analysis of coupled dynamical systems.}
    \label{fig:stucture2}
\end{figure}
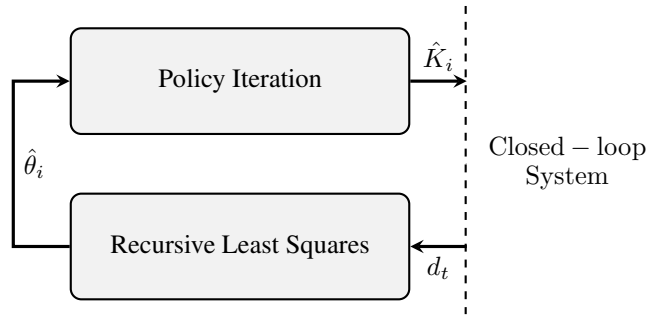

Before stating the main result, we introduce two assumptions.
\begin{assumption}\label{assumption 1}
   The estimates $\left(\hat{A}_i,\hat{B}_i\right)$ obtained from RLS are stabilizable $\forall i\in \mathbb{Z}_+$. 
\end{assumption}
We note that the employed RLS scheme does not provide any a priori guarantee on the satisfaction of this assumption. However, this condition can be checked before using the estimated model in the PI routine, and thus in practice, whenever this condition is violated, the system is excited and RLS is applied
until a stabilizable estimate is obtained. 
Because the true system $(A,B)$ is stabilizable, under the assumption of a locally persistent data sequence, we can also prove the existence of an index $i_\mathrm{st}$ such that $(\hat{A}_i,\hat{B}_i)$ obtained from RLS is stabilizable for all $i\geq i_\mathrm{st}$. This is formally stated next.
\begin{theorem}\label{theorem5}
   For any initialization $\left(\hat{A}_0,\hat{B}_0\right)$ and stabilizable pair $(A,B)$, given a locally persistent sequence $\{D_i\}$ with any lower bound $\alpha>0$ and persistency window $N\in \mathbb{Z}_{++}$ and $M\in \mathbb{Z}_{++}$, there always exists an index $i_\mathrm{st}$ such that $\forall i\geq i_\mathrm{st}$, the estimate obtained from RLS $\left(\hat{A}_i,\hat{B}_i\right)$, is stabilizable. 
\end{theorem}
The proof of Theorem \ref{theorem5} is given in Appendix \ref{prooftheorem5}. This theorem ensures that there always exists an $i_\mathrm{st}$ such that after $i_\mathrm{st}$, Assumption \ref{assumption 1} is always satisfied. 

\begin{assumption}\label{assumption 2}
   Given a stabilizable estimate $\left(\hat{A}_i,\hat{B}_i\right)$, we assume that $\hat{P}_i \succeq P^*_{\left(\hat{A}_i,\hat{B}_i\right)}$ $\forall i\in \mathbb{Z}_+$, where $\hat{P}_i$ is obtained via \eqref{IPI2a} and $P^*_{\left(\hat{A}_i,\hat{B}_i\right)}$ is the quadratic kernel of the value function associated with $\left(\hat{A}_i,\hat{B}_i\right)$ and is calculated by solving 
   \eqref{DARE}.   
\end{assumption}
Assumption \ref{assumption 2} guarantees that item 1 and item 2 in Theorem \ref{theorem1} hold, which is convenient for the subsequent analysis. Moreover, if $\hat{P}_i \succeq P^*_{\left(\hat{A}_i,\hat{B}_i\right)}$ and we iteratively refine $\hat{P}_i$ via Algorithm \ref{Algo1} using $\left(\hat{A}_i,\hat{B}_i\right)$, then $\hat{P}_i$ converges exponentially to $P^*_{\left(\hat{A}_i,\hat{B}_i\right)}$. 
On the contrary, if $\hat{P}_i \nsucceq P^*_{\left(\hat{A}_i,\hat{B}_i\right)}$, there are no guarantees that $\hat{K}_i$ stabilizes $\left(\hat{A}_i,\hat{B}_i\right)$ and the subsequent policy evaluation might be infeasible. In practice, one only needs to have that $\hat{K}_i$ stabilizes the associated estimate, and this can be checked directly. Heuristics to achieve 
Assumption \ref{assumption 2} include again  
choosing $\tau_\mathrm{IPI}$ sufficiently large, or re-initializing $\hat{K}_i$ with a gain that stabilizes the current estimated system dynamic $\left(\hat{A}_i,\hat{B}_i\right)$. Relaxing this requirement with tighter sufficient conditions is left for future work.

We state next the main result of the section.

\begin{theorem}\label{theorem6}
  If Assumption \ref{assumption 1} and Assumption \ref{assumption 2} are satisfied, then the coupled recursive least squares and policy iteration system formulated by \eqref{IPI2a}-\eqref{IPI2d} admits the equivalent dynamical system representation: 
  \begin{subequations}\label{IPI2}
  \begin{align}
  \hat{\theta}_{i+1} &= \left(\hat{\theta}_{i}\left(\sum\limits_{t=1}^{i\tau_\mathrm{IPI}}d_td_t^\top \right)+ \sum\limits_{t=i\tau_\mathrm{IPI}+1}^{(i+1)\tau_\mathrm{IPI}}x_{t+1}d_t^\top \right)\left(\sum\limits_{t=1}^{(i+1)\tau_\mathrm{IPI}}d_td_t^\top \right) ^{-1},\label{IPI21c}\\
  \hat{P}_{i+1}&=\mathcal{L}^{-1}_{\left(\hat{A}_i,\hat{B}_i,\hat{P}_i\right)}\left(Q+\alpha_i^\top\left(\hat{P}_i\right)\beta_i^{-1}\left(\hat{P}_i\right)R\beta_i^{-1}\left(\hat{P}_i\right)\alpha_i\left(\hat{P}_i\right)\right),\label{IPI21d}
  \end{align}
\end{subequations}
where $\alpha_i\left(\hat{P}_i\right):=\hat{B}_i^\top \hat{P}_i\hat{A}_i$ and $\beta_i\left(\hat{P}_i\right)=\beta_i^\top\left(\hat{P}_i\right):=R+\hat{B}_i^\top \hat{P}_i\hat{B}_i$. Then, for any positive integer $\tau_\mathrm{IPI} \in \mathbb{Z}_{++}$ the estimate $\hat{P}_i$ and $\hat{\theta}_i$ satisfy the following relationships: 
  \begin{subequations}\label{ISSCA_TOT}
  \begin{align}
   \left \lVert \hat{P}_{i}- P^*\right \rVert_F &\leq \beta\left(\left \lVert \hat{P}_{0}- P^*\right \rVert_F,i\right)+\gamma\left(\left \lVert \Omega\right \rVert_{\infty}\right),\label{ISSCA}\\
   \left \lVert \hat{\theta}_i-\theta\right \rVert_F&\leq \Delta\theta^\mathrm{Upper}_i\leq f\left(\left \lVert \hat{\theta}_0-\theta \right \rVert_F,i\right)+g\left(\left \lVert j^\mathrm{non}\right \rVert_{\infty}\right),\label{ISSCA2}
  \end{align}
\end{subequations}
where: $\beta\left(\cdot, \cdot \right):=c^i\left \lVert \hat{P}^{0}- P^*\right \rVert_F$ is a $\mathcal{KL}$ function with $c \in (0,1)$; $\gamma\left(\left \lVert \Omega\right \rVert_{\infty}\right):=\frac{1}{1-c}\left \lVert \Omega\right \rVert_{\infty}$ is a $\mathcal{K}$ function; $\Omega$ is a scalar sequence with $\Omega_i:=\sigma_i\Delta\theta^\mathrm{Upper}_i$ and $\sigma_i:=\sum\limits_{k=0}^{9}\rho_k\left(\Delta\theta^\mathrm{Upper}_i\right)^k$; and the detailed expressions of $\rho_k >0$ are provided in Appendix \ref{App2}. The functions $f(\cdot)$ and $g(\cdot)$ are given in Theorem \ref{theorem3}. 
\end{theorem}

The proof of Theorem \ref{theorem6} can be found in Appendix \ref{App2}.
\begin{remark}{(ISS-type interpretation of \eqref{ISSCA_TOT} in Theorem \ref{theorem6})}\\
A few comments are given in order to better interpret the result of Theorem \ref{theorem6}. For this, we would like to briefly recall input-to-state stability (ISS), which is a notion of stability widely employed to analyze the stability of nonlinear systems subject to external disturbances \cite{Sontag2008} \cite{JIANG2001857}. Given a nonlinear system $y_{t+1}=f(y_t,w_t)$, the system is said to be input to state stable if there exist $\beta_{\mathrm{ISS}}\in\mathcal{KL}$ and $\gamma_{\mathrm{ISS}}\in\mathcal{K}$ such that
\begin{equation}\label{standardISS}
    \lVert y_t \rVert \leq \beta_{\mathrm{ISS}}\left(\lVert y_0 \rVert,t\right)+\gamma_{\mathrm{ISS}}\left(\lVert w \rVert_{\infty}\right).
\end{equation}
As can be seen from the equation above, if a system is ISS, then it is globally asymptotically stable in the absence of external inputs and its state is bounded by a function of the size of the external inputs.

Let us go back now to equation \eqref{ISSCA}. 
Using boundedness of $\Delta\theta^\mathrm{Upper}_i$ (see Lemma \ref{theorem21} in the Appendix \ref{App11}) and thus of $\sigma_i$, we can rewrite it equivalently as:
     \begin{equation}\label{reformulated}
         \left \lVert \hat{P}_{i}- P^*\right \rVert_F \leq \beta\left(\left \lVert \hat{P}_{0}- P^*\right \rVert_F,i\right)+\bar{\gamma}\left(\left \lVert \Delta\theta^\mathrm{Upper}\right \rVert_{\infty}\right).
     \end{equation}
where $\bar{\gamma}\left(\lVert \Delta\theta^\mathrm{Upper}\rVert_{\infty}\right):=\frac{\sigma_0}{1-c}\lVert \Delta\theta^\mathrm{Upper}\rVert_{\infty}$. Eq.\eqref{reformulated} says that the policy iteration system formulated by \eqref{IPI21d} is input-to-state stable with respect to the sequence $\Delta\theta^\mathrm{Upper}$ representing the RLS estimation error. 
     
Let us now turn our attention to equation \eqref{ISSCA2}. The function $f(x,i)$ is not a $\mathcal{KL}$ function unless $N_{\max}=1$ because of the floor function. However, note that $f(x,i)$ is always non-increasing with respect to $i$, $f(x,i)\rightarrow 0$ when $i\rightarrow\infty$ and $f(x,i)\rightarrow0$ when $x\rightarrow0$. The function $g(x)$ is not a $\mathcal{K}$ function because $j_i^{\mathrm{non}} \in \mathbb{Z}_+$, but satisfies $g(x)\rightarrow0$ when $x\rightarrow0$. 
Therefore, we can qualitatively interpret \eqref{ISSCA2} as an ISS-type stability result. 

The interesting aspect of this informal interpretation is that we can view the sequences associated with estimation error and lack of persistent excitation as disturbance terms acting on the PI and RLS dynamics, respectively. 
 
\end{remark}

Building on this interpretation, we finally provide the following two corollaries.
\begin{corollary}\label{Coro3}
Given the assumptions and notations of Theorem \ref{theorem6}, and an episode $i_{\mathrm{re}}>1$, the following holds 
    \begin{equation}\label{reISSCA3}
    \left \lVert \hat{P}_{i}- P^*\right \rVert_F \leq \beta\left(\left\lVert \hat{P}_{i_{\mathrm{re}}}- P^*\right \rVert_F,i-i_{\mathrm{re}}\right)+\frac{\sigma_{i_{\mathrm{re}}}}{1-c} \Delta\theta^\mathrm{Upper}_{i_{\mathrm{re}}},\quad i\geq i_{\mathrm{re}}.
\end{equation}
\end{corollary}
Because $\sigma_{i_{\mathrm{re}}}\leq \sigma_{0}$, we can see also analytically that a faster convergence of $\Delta\theta^\mathrm{Upper}_i$ contributes to a faster convergence of $\hat{P}_{i}$.
\begin{corollary}\label{Coro1}
  Under the assumptions and notations of Theorem \ref{theorem6}, if the data sequence $\{d_t\}$ is locally persistent with any $N,M,\alpha$, then $\mathop{lim}\limits_{i\rightarrow\infty}\hat{\theta}_i=\theta $,$\mathop{lim}\limits_{i\rightarrow\infty}\hat{P_i}=P^*$ and $\mathop{lim}\limits_{i\rightarrow\infty}\hat{K_i}=K^*$. 
\end{corollary}
Corollary \ref{Coro1}, which is proven in Appendix \ref{App3}, shows that IPI asymptotically finds the optimal gain $K^*$. We remark that this holds for any value of $\tau_\mathrm{IPI}$. 

When the sequence is not locally persistent, Theorem \ref{theorem6} is still informative and allows a bound on the distance $\left \lVert \hat{P}_{i}- P^*\right \rVert_F$ to be quantified as:
\begin{equation}\label{composition_error}
    \left \lVert \hat{P}_{i}- P^*\right \rVert_F\leq \beta\left(\left \lVert \hat{P}_{0}- P^*\right \rVert_F,i\right)+\bar{\gamma}\left(f\left(\left \lVert \hat{\theta}_0-\theta \right \rVert_F,i\right)+g\left(\left \lVert j^\mathrm{non}\right \rVert_{\infty}\right)\right).
\end{equation}

\section{Direct Data-driven Policy Iteration}\label{sec:DDD}

This Section investigates a possible approach to perform policy iteration using data but without having to explicitly identify a model. 
This approach, referred to here as DPI to mark its direct nature, is a modified version of the algorithm recently presented in \cite{9691800} to address identifiability issues arising in noise-free scenarios. 
The DPI algorithm has the same episodic structure introduced in Section \ref{sec:IDD} with reference to IPI, and therefore we also use here the two indices timesteps $t$ and episode step $i$ with the same meaning. The key difference is that the data collected in each episode, of length $\tau_\mathrm{DPI}$, are used here to learn the parameters required in the policy iteration update involved in the PI steps.

\subsection{Preliminaries of model-free policy iteration}
In this section, we review the necessary background material from \cite{9691800}. At a generic episode $i \geq 1$, the DPI procedure involves the two steps described next.
\subsubsection{Model-free policy evaluation}\label{MFPEsection}
In order to perform the policy evaluation step in Algorithm \ref{Algo1} in a model-free fashion, the key observation is that \eqref{Bellman}, used to explain the basic idea of PI, does not require the system dynamics. Introducing $vecs(P_i)$, where $x_t$ is the system state at timestep $t$ and $P_i$ is the quadratic kernel at episode $i$, this equation can be rewritten as follows: 
\begin{equation}\label{PEbellmanequation2}
  {\underbrace{\left(vecv(x_t)-vecv(x_{t+1})\right)}_{=:\Phi_t} }^\top vecs(P_i)=x_t^\top Qx_t+u_t^\top Ru_t=r(x_t,K_ix_t)
\end{equation}
Estimating $vecs(P_i)$ requires solving the linear equation \eqref{PEbellmanequation2}, which consists of $\frac{n_x(n_x+1)}{2}$ unknown entries defining the upper-triangle part of a symmetric matrix. 
To this end, the system is excited using the control input $u_t=K_ix_t$ for at least $\tau_\mathrm{PE}=\frac{n_x(n_x+1)}{2}$ instances. 
If the matrix $\left(\sum\limits_{k=t}^{t+\tau_\mathrm{PE}-1}\Phi_k\Phi_k^\top\right)$ is invertible,
then $vecs(P_i)$ can be exactly estimated from the data $\{x_t,...,x_{t+\tau_\mathrm{PE}}\}$ as follows:
\begin{equation}\label{Leastsquare1}
\begin{split}
   vecs({P}_i)=\left(\sum\limits_{k=t}^{t+\tau_\mathrm{PE}-1}\Phi_k\Phi_k^\top\right)^{-1}\left( \sum\limits_{k=t}^{t+\tau_\mathrm{PE}-1}\Phi_k r(x_k,K_ix_k) \right).
\end{split}
\end{equation}
\begin{remark}\label{remark2}
   This method is only applicable to a system subject to process noise, as for example in $x_{t+1}=Ax_t+Bu_t+w_t$. In noise-free cases (or cases with a very low amount of process noise), introducing the simple linear feedback results in the evolution $x_{t+1}=(A+BK_i)x_t$ which determines a lack of identifiability of $P_i$. The discussion on decoupling the data dependency and refining this method to resolve this issue will be presented in the following subsection \ref{Prodecure:DDD}.
\end{remark}
\subsubsection{Model-free policy improvement}\label{2PI}
The objective in this step is to find the updated policy $K_{i+1}$ based on the evaluation $P_i$. The system is excited by $u_t=K_ix_t+\eta_t$, where $\eta_t \sim \mathcal{N}(0, W_\eta)$ serves as additional excitation and must be present. The reason for the need of introducing this additional excitation in the policy will be evident soon. The corresponding value-based Bellman equation can be expressed as:
\begin{equation}\label{PIbellmanequation1}
  (x_{t+1}-B\eta_t)^\top P_i(x_{t+1}-B\eta_k)=x_t^\top(P_i-Q-K_i^\top RK_{i})x_t.
\end{equation}
This leads to:
\begin{equation}\label{PIbellmanequation2}
  \begin{split}
2x_t^\top {A^\top P_iB}\eta_t+(2K_ix_t+\eta_t){B^\top P_iB}(u_t-K_ix_t)&=x_t^\top (Q+K_i^\top RK_{i}-P_i)x_t+x_{t+1}^\top P_ix_{t+1},\\
{\underbrace{\left [\begin{array}{c}
       2x_t \otimes \eta_t \\
       vecv(u_t)-vecv(K_ix_t) 
     \end{array}\right ]}_{=:\Gamma_t}}^\top \underbrace{\left [\begin{array}{c}
       vec(B^\top P_iA) \\
       vecs(B^\top P_iB) 
     \end{array}\right ]}_{=:\xi_i}&=\underbrace{x_t^\top (Q+K_i^\top RK_{i}-P_i)x_t+x_{t+1}^\top P_ix_{t+1}}_{=:c_t}.
  \end{split}
\end{equation}
Eq. \eqref{PIbellmanequation2} shows that introducing $\eta_t$ is critical to ensure that $B^\top P_iB, B^\top P_iA$ can be estimated, otherwise $\Gamma_t$ is always zero matrix. 
The vector $vecs(B^\top P_iB)$ consists of $\frac{n_u(n_u+1)}{2}$ unknown entries (because of symmetry) whereas the vector $vec(B^\top P_iA)$ presents $n_xn_u$ unknown. If the matrix $\left(\sum\limits_{k=t}^{t+\tau_\mathrm{PI}-1}\Gamma_t\Gamma_t^\top\right)$ is invertible,
then $\xi_i$ can be obtained as follows:
\begin{equation}\label{Leastsquare2}
  \xi_i=\left( \sum\limits_{k=t}^{t+\tau_\mathrm{PI}-1}\Gamma_k\Gamma_k^\top  \right)^{-1}\left( \sum\limits_{k=t}^{t+\tau_\mathrm{PI}-1}\Gamma_k c_k \right)
\end{equation} 
After estimating $\xi_i$, the gain can be updated with $K_{i+1}=-(R+B^\top P_iB)^{-1}B^\top P_iA$. 
\begin{remark}
 If the number of minimum samples discussed before are fulfilled and the sequences $\{\Phi_t\}$ and $\{\Gamma_t\}$ are globally persistent with window length $N$ defined earlier, $P_i$ and $\xi_i$ are estimated exactly at each step and thus the properties of model-based policy iteration in Theorem \ref{theorem1} apply here as well.  
\end{remark}

\subsection{Algorithm definition}\label{Prodecure:DDD}
As mentioned in Remark \ref{remark2}, an issue in \cite{9691800} preventing its use in a noise-free scenario is the loss of identifiability of $P_i$. To overcome it, one solution is to add a signal $\eta_t$ targeted to decouple the dependence across data points. The issue with using an additional excitation $\eta_t$ for the policy evaluation step is that the evaluated policy is the affine policy $\hat{K}_ix_t+\eta_t$ and not just the linear state-feedback one, thus introducing an error. To solve this problem, we introduce an average system. 

The average system is formulated as follows: when $t$ is odd, $\eta_t$ is randomly selected from a Gaussian distribution, and when $t$ is even, $\eta_t$ is selected as $\eta_t=-\eta_{t-1}$. The state evolutions at two consecutive times when $t$ is odd are given as follows:
\begin{subequations}\label{MFPI1}
   \begin{align}
     x_{t+1} &= Ax_{t}+Bu_t=Ax_{t}+B (\hat{K}_ix_{t}+\eta_t),\label{MFPIa} \\
     x_{t+2}= Ax_{t+1}+Bu_t=&Ax_{t+1}+B (\hat{K}_ix_{t+1}+\eta_{t+1})=Ax_{t+1}+B (\hat{K}_ix_{t+1}-\eta_{t}).\label{MFPIb}
   \end{align}
\end{subequations}
By summing up the state evolutions, we obtain the average system dynamics:
\begin{equation}\label{MFPI2}
   (x_{t+2}+x_{t+1})=A(x_{t+1}+x_{t})+B\hat{K}_i(x_{t+1}+x_{t}).
\end{equation}
Because $t$ is odd, we can define $t:=2k-1$ with $k\in \mathbb{Z}_+$. This allows the use of the data pairs $\{x_{2k}+x_{2k-1}\}$ together with $\{x_{2k+1}+x_{2k}\}$ to evaluate the gain $\hat{K}_i$, as described in Section \ref{MFPEsection}. Because the policy evaluation is executed based on the average system (data pairs), each episode must consist of an even number of data points. The stage cost for the average system is defined as:
\begin{equation}\label{MFPI3}
   R_k:=(x_{2k-1}+x_{2k})^\top Q(x_{2k-1}+x_{2k})+(x_{2k-1}+x_{2k})^\top\hat{K}_{i}^{\top}R\hat{K}_i(x_{2k-1}+x_{2k}).
\end{equation}
This leads to a newly formulated value-based Bellman equation:
\begin{equation}\label{MFPI4}
   (x_{2k}+x_{2k+1})^\top\hat{P}_i(x_{2k}+x_{2k+1})+R_k=(x_{2k-1}+x_{2k})^\top\hat{P}_i(x_{2k-1}+x_{2k}).
\end{equation}
Consequently, we can obtain $\hat{P}_i$ by solving the following equation:
\begin{equation}\label{MFPI5}
  \begin{aligned}
     {\underbrace{\left(vecv(x_{2k-1}+x_{2k})-vecv(x_{2k}+x_{2k+1})\right)}_{=:\phi_k}}^\top vecs(\hat{P}_i)=R_k.
  \end{aligned}
\end{equation}
We collect the data from time index $t$ to $t+\tau_\mathrm{DPI}-1$, constituting an episode with length $\tau_\mathrm{DPI}$. The collected data provide $\frac{\tau_\mathrm{DPI}}{2}$ data pairs. If the matrix $\left( \sum\limits_{k=\frac{t+1}{2}}^{\frac{t+1}{2}+\frac{\tau_\mathrm{DPI}}{2}}\phi_k\phi_k^\top \right)$ is invertible,
then the solution to equation \eqref{MFPI5} is given by:
\begin{equation}\label{MFPI6}
  vecs(\hat{P}_i)=\left( \sum\limits_{k=\frac{t+1}{2}}^{\frac{t+1}{2}+\frac{\tau_\mathrm{DPI}}{2}-1}\phi_k\phi_k^\top \right)^{-1}\left( \sum\limits_{k=\frac{t+1}{2}}^{\frac{t+1}{2}+\frac{\tau_\mathrm{DPI}}{2}-1}\phi_k R_k \right).
\end{equation}
The policy improvement step does not require modifications and thus is taken identically to Section \ref{2PI}. 
The overall procedure of DPI is summarized in Algorithm \ref{Algo3} and the invertibility conditions in \eqref{Leastsquare2} and \eqref{MFPI6} are discussed in the following section.\\
\begin{algorithm}
  \caption{Direct data-driven policy iteration.}\label{Algo3}
  \begin{algorithmic}
      \Require stabilizing policy gain $K_1$ 
      \For{$i=1,...,\infty$} (episode counter)
        \State \For{$j=1,...,\tau_\mathrm{DPI}$} (timestep counter inside episode $i$)
        \State \State \textbf{Excite the system with $u_{\tau_\mathrm{DPI}(i-1)+2j-1}=\hat{K}_ix_{\tau_\mathrm{DPI}(i-1)+2j-1}+\eta_{\tau_\mathrm{DPI}(i-1)+2j-1}$, $\eta_t\sim \mathcal{N}(0,W_\eta)$}
        \State \State \textbf{Collect the data $D_{PI} \gets (x_{\tau_\mathrm{DPI}(i-1)+2j-1},u_{\tau_\mathrm{DPI}(i-1)+2j-1},x_{\tau_\mathrm{DPI}(i-1)+2j})$}
        \State \State \textbf{Excite the system with $u_{\tau_\mathrm{DPI}(i-1)+2j}=\hat{K}_ix_{\tau_\mathrm{DPI}(i-1)+2j}-\eta_{\tau_\mathrm{DPI}(i-1)+2j-1}$  }
        \State \State \textbf{Collect the data $D_{PI} \gets (x_{\tau_\mathrm{DPI}(i-1)+2j},u_{\tau_\mathrm{DPI}(i-1)+2j},x_{\tau_\mathrm{DPI}(i-1)+2j+1})$}
        \State \State \textbf{Collect the data $D_{PE} \gets (x_{\tau_\mathrm{DPI}(i-1)+2j-1}+x_{\tau_\mathrm{DPI}(i-1)+2j},u_{\tau_\mathrm{DPI}(i-1)+2j-1}+u_{\tau_\mathrm{DPI}(i-1)+2j},x_{\tau_\mathrm{DPI}(i-1)+2j}+x_{\tau_\mathrm{DPI}(i-1)+2j+1})$}
        \State \EndFor
        \State \textbf{Policy Evaluation: find $\hat{P}_{i}$} 
        \State \textbf{Calculate $\hat{P}_{i}$ via \eqref{MFPI6} with $D_{PE}$}
        \State \textbf{Policy Improvement: update gain $\hat{K}_{i+1}$}
        \State \textbf{Calculate $B^\top P_iB, B^\top P_iA$ via \eqref{Leastsquare2} with $D_{PI}$}
        \State $\hat{K}_{i+1}=-(R+B^\top P_iB)^{-1}B^\top P_iA$
      \EndFor
  \end{algorithmic}
\end{algorithm}

\begin{remark}{(Notation of $\hat{P}_i$ in DPI)}\label{remark_notation_pi}\\
    In the Section \ref{Prodecure:DDD} and also the subsequent sections, the output of the policy evaluation in DPI is denoted as $\hat{P}_i$. When the inverse operation in the \eqref{MFPI6} exists, $\hat{P}_i$ coincides with $P_i$ from the model-based policy iteration with the same initialization $K_1$, because of the noise-free scenarios, i.e. $\hat{P}_i=P_i$. However, we adopt here the notation $\hat{P}_i$ to emphasize that this is an iterated quantity inside the corresponding learning-based algorithm, as it is the case for the variable $\hat{P}_i$ in IPI. 
\end{remark}
\subsection{Conditions for convergence of DPI}\label{DPIConvergenceAnalysis}
First, we analyze the minimum number of samples required in each episode. The symmetric matrix $\hat{P}_i$ requiring a minimum number of $\frac{n_x(n_x+1)}{2}$ data pairs for its exact computation. As a result, $n_x(n_x+1)$ data points are necessary during the policy evaluation phase. In the policy improvement phase, $B^\top P_iB$ and $B^\top P_iA$ are estimated. We note that $B^\top P_iB$ is also a symmetric matrix. Hence, a total of $\frac{n_u(n_u+1)}{2}+n_un_x$ data points are necessary for the policy improvement. Crucially, the same data points are employed for both the policy evaluation and policy improvement simultaneously. Thus, the minimum number of required samples for each episode, which influences the choice of $\tau_\mathrm{DPI}$, is determined by $\max[n_x(n_x+1),\frac{n_u(n_u+1)}{2}+n_un_x]$, i.e. $\tau_\mathrm{DPI}\geq\max[n_x(n_x+1),\frac{n_u(n_u+1)}{2}+n_un_x]$.

We then turn our attention to the invertibility conditions in \eqref{Leastsquare2} and \eqref{MFPI6}, which are related to the \emph{informativity} of the data sequence used there. When the required number of samples in each episode are provided and in addition the data sequences $\phi$ and $\Gamma$ are locally persistent with $N=M=\frac{\tau_{\mathrm{DPI}}}{2}$ and $N=M=\tau_{\mathrm{DPI}}$,
respectively, then from Definition \ref{def2} it holds 
\begin{equation}\label{DPIcondition}
  \left( \sum\limits_{k=\frac{t+1}{2}}^{\frac{t+1}{2}+\frac{\tau_\mathrm{DPI}}{2}-1}\phi_k\phi_k^\top \right)\succeq0, \left( \sum\limits_{k=t}^{t+\tau_\mathrm{DPI}-1}\Gamma_k\Gamma_k^\top  \right)\succeq0, \forall t=i\tau_{\mathrm{DPI}}+1, \forall i\in \mathbb{Z}_{+}.
\end{equation}
Then the inverse operations in \eqref{Leastsquare2} and \eqref{MFPI6} can be performed in each episode and the estimations of $\hat{P}_i$, $B^\top P_iB$ and $B^\top P_iA$ are exact. In this case, as noted in Remark \ref{remark_notation_pi}, the algorithm produces the same iterates of classic PI with model knowledge. Therefore, invoking Theorem \ref{theorem1}, we can conclude that $\hat{P}_i$ and $\hat{K}_i$ converge to $P^*$ and $K^*$, respectively.

When the length of the episode is insufficient to meet the overall required number of samples per episode, Eq. \eqref{MFPI5} and Eq. \eqref{PIbellmanequation2} become underdetermined. This implies that $\hat{P}_i$, $B^\top P_iB$, and $B^\top P_iA$ cannot be precisely solved, leading to the divergence of $\hat{K}_i$ and $\hat{P}_i$.

\section{Discussion of IPI and DPI}\label{sec:discussion}

In this section, we compare the salient features of the analyses performed to IPI and DPI in Section \ref{sec:IDD} and \ref{sec:DDD}, respectively. A summary of the key aspects is presented in Table \ref{table} at the end of this section. 
\subsection{Selection of episode length}\label{selectiontau}

The selection of $\tau_\mathrm{IPI}$ has two major implications. First, it can have a beneficial effect on the achievement of Assumptions \ref{assumption 1} and \ref{assumption 2}, as already discussed therein. On the other hand, it is a trade-off parameter between the rate at which the system dynamics are learned and how often the policy iteration is performed. The theoretical minimum value for $\tau_\mathrm{IPI}$ is $1$, and its choice is a system-dependent compromise among the aforementioned aspects. The selection of $\tau_\mathrm{DPI}$ must satisfy the minimum number of required samples of DPI, i.e. $\tau_\mathrm{DPI}\geq\mathrm{max}[n_x(n_x+1),\frac{n_u(n_u+1)}{2}+n_un_x]$ and $\tau_\mathrm{DPI}$ must be even. In all other cases, convergence cannot be guaranteed.
In Table \ref{table} we refer to the minimum number of requires samples in each episode of the two methods with \textit{minimum episode length}.

\subsection{Persistency condition}\label{Informativity}
Corollary \ref{Coro1} guarantees convergence of IPI when the data sequence $\{d_t\}$ is locally persistent with any $N,M,\alpha$. 
This can be ensured by choosing an appropriate dithering signal $e_t$, e.g. $e_t \sim \mathcal{N}(0, W_e)$ for appropriate $W_e \succ 0$. Even when the local persistency requirement is not met, recursive least squares can still be applied because as mentioned in Remark \ref{rank1update} the inverse involved in its step always exists and is in fact never explicitly computed. Moreover, Theorem \ref{theorem6} provides for IPI computable bounds on the error between estimates and true values due to a lack of excitation in the sequence. These bounds explicitly point out the contribution to the error due to the non-exciting part of the sequence (see terms depending on the function $g$ in \eqref{composition_error}).\\
In contrast, convergence of the DPI algorithm requires exact computation of $\hat{P}_i$, $B^\top P_iB$ and $B^\top P_iA$ in each episode, which in turn requires that the invertibility conditions in \eqref{Leastsquare2} and \eqref{MFPI6} hold. As discussed in Section \ref{DPIConvergenceAnalysis}, this requires a minimum number of samples (discussed in Section \ref{DPIConvergenceAnalysis}) and a sufficient richness of the data. Precisely, the data sequence $\phi$ (used to determine $\hat{P}_i$) must be locally persistent with $N=M=\frac{\tau_\mathrm{DPI}}{2}$, whereas the data sequence $\Gamma$ (used to determine $B^\top P_iB$ and $B^\top P_iA$) must be locally persistent with $N=M=\tau_\mathrm{DPI}$ (recall Section \ref{DPIConvergenceAnalysis}).  
One specific issue of DPI, as opposed to IPI, is that, because of the nonlinear vector operation $vecs(\cdot)$ affecting $\phi$ and $\Gamma$, local persistency of the data sequences cannot be guaranteed a-priori by design of the dithering signal. 
Because local persistency conditions are necessary to establish convergence, establishing a-priori guarantees for DPI is still an open problem. 
The content of this subsection is summarized in Table \ref{table} with the items \textit{required persistency condition} and \textit{guarantee on persistency condition}.

\subsection{Convergence rate}\label{IPI_DPI_rates}
If the excitation conditions discussed in Section \ref{Informativity} are met in \emph{each} episode, then DPI follows the convergence rate exhibited by policy iteration using the true model. Under this condition, the upper bound on the iterate error of DPI, denoted as $f_\mathrm{DPI}^{\mathrm{Upper}}(i_\mathrm{DPI})$,  can then be derived using Theorem \ref{theorem4} as follows:
\begin{equation}\label{convergenceDPI}
  \begin{split}
    \lVert \hat{P}_{i_\mathrm{DPI}}-P^* \rVert_F &\leq c\lVert \hat{P}_{i_\mathrm{DPI}-1}-P^* \rVert_F \leq...\leq c^{{i_\mathrm{DPI}}}\lVert \hat{P}_{0}-P^* \rVert_F=:f_\mathrm{DPI}^{\mathrm{Upper}}(i_\mathrm{DPI}),
  \end{split}
\end{equation}
where $i_\mathrm{DPI}$ is the episode index of DPI. Defining the episode index of IPI as $i_\mathrm{IPI}$, the upper bound on the iterate error of IPI denoted as $f_\mathrm{IPI}^{\mathrm{Upper}}(i_\mathrm{IPI})$, outlined in Theorem \ref{theorem6}, can be obtained as follows:
\begin{equation}\label{convergenceIPI}
  \lVert \hat{P}_{i_\mathrm{IPI}}- P^*\rVert_F \leq c^{i_\mathrm{IPI}}\lVert \hat{P}_{0}- P^*\rVert_F+\frac{\sigma_{0}}{1-c}\lVert \Delta\theta^\mathrm{Upper}\rVert_{\infty}=:f_\mathrm{IPI}^{\mathrm{Upper}}(i_\mathrm{IPI}).
\end{equation} 
We can then draw a comparison between the upper bounds on the error of DPI and IPI starting from an identical initialization $\hat{P_0}$ at 
the same episode counter $i_\mathrm{DPI}=i_\mathrm{IPI}$:
\begin{equation}\label{Compare1}
  f_\mathrm{DPI}^{\mathrm{Upper}}(i_\mathrm{DPI})=c^{{i_\mathrm{DPI}}}\lVert \hat{P}_{0}-P^* \rVert_F \leq c^{i_\mathrm{IPI}}\lVert \hat{P}_{0}- P^*\rVert_F+\frac{\sigma_{0}}{1-c}\lVert \Delta\theta^\mathrm{Upper}\rVert_{\infty}=f_\mathrm{IPI}^{\mathrm{Upper}}(i_\mathrm{IPI}).
\end{equation}
where the inequality comes from the fact that, as seen in \eqref{convergenceIPI}, $f_\mathrm{IPI}^{\mathrm{Upper}}(i_\mathrm{IPI})$ has an additional non-negative term due to the model mismatch. As a result, at the same episode $f_\mathrm{DPI}^{\mathrm{Upper}}(i_\mathrm{DPI})$ is always smaller than $f_\mathrm{IPI}^{\mathrm{Upper}}(i_\mathrm{IPI})$. 
However, it is essential to observe that the minimum episode length required by IPI and DPI, discussed in the previous section, is markedly different. While for IPI the theoretical minimum number of samples is $1$, for DPI this depends polynomially on the state and input sizes. This means that $\tau_\mathrm{IPI}$ can be much smaller than $\tau_\mathrm{DPI}$, allowing more episodes to take place within the same number of total timesteps. If we compare the two methods given the same fixed time budget, that is
$\tau_\mathrm{DPI}i_\mathrm{DPI}=\tau_\mathrm{IPI}i_\mathrm{IPI}$, then IPI can perform a much larger number of episode.

The explanation for this much better \emph{sample complexity} of IPI, when compared to DPI, has to do with the different learning objectives of the algorithms. 
Precisely, IPI utilizes data only to learn the system dynamics $(A,B)$, while DPI employs data to learn in each episode the value function and the improved policy $(P_i,K_i)$. Whereas the system dynamics $(A,B)$ is time-invariant and thus all the available data contribute towards its estimation, $(P_i,K_i)$ are episode-related quantities and hence the increase of information attained by each data point is limited to the current episode.
From this perspective, identifying a model allows a more efficient use of the data to be made, because effectively it decreases the number of unknowns that should be estimated from data. The information provided by this subsection is summarized and labeled as \textit{convergence rate} in Table \ref{table}.

\subsection{Excitation policy}\label{Selection}

In IPI, the control input is expressed as $u_t=\hat{K}_ix_t+e_t$, while in DPI, the control input is given as $u_t=\hat{K}_ix_t+\eta_t$. However, there is a notable distinction in the selection of the individual terms of this policy. In DPI, $\eta_t$ needs to satisfy specific conditions to form data pairs, and $\hat{K}_i$ must be the policy gain obtained from the policy improvement step of the previous episode, i.e. data should be "on-policy". On the other hand, IPI allows for a more flexible selection of $\hat{K}_i$. 
The use of a model to perform the steps in PI allows the controller design and online data collection to be decoupled, because the collected data are used to learn a model and not the evaluation and improvement of the controller. 
In this work, the choice of controller $u_t=\hat{K}_ix_t+e_t$ used during IPI is only driven by the aim to minimize the cost function \eqref{Cost} also during learning.
However, IPI does not require "on-policy" data, and thus any controller, e.g. any pre-stabilizing gain $\hat{K}$, would be an acceptable solution. From this perspective, the availability of a model provides for this problem greater flexibility in selecting the online controller. In this work, we do not specify the selection of $e_t$ in IPI, while our analysis holds for any selection of $\hat{K}_i$ and $e_t$. In the following simulations, we select $e_t$ as a Gaussian distributed signal to ensure persistency and also $e_t=0$ to verify the performance of IPI with a non-persistent sequence. Understanding how to best design $e_t$ based on online available information such as $\hat{P}_{i}$ or $H_i$ is a natural follow-up to this work. The content of this subsection is summarized in Table \ref{table} with the items \textit{selection of excitation}. 

\subsection{Closed-loop stability of the plant and algorithmic dynamics}
In Section \ref{IPIConvergenceAnalysis}, we analyze the algorithmic closed-loop dynamic formed by RLS and PI. This is also illustrated with the RLS and PI blocks connected by the dashed lines in Figure \ref{fig:stucture}. Whereas Theorem \ref{theorem6}, under certain assumptions, shows a bounded response for this interconnection, we do not discuss the closed-loop stability of the interconnection with the physical plant, depicted by the system dynamics and controller blocks connected with solid lines. For the "off-policy" case where IPI runs with a given pre-stabilizing gain $\hat{K}$, this analysis is trivial and we could conclude closed-loop stability of the physical-algorithmic system because the plant is effectively decoupled from the algorithm. In the "on-policy" case, this is an interesting question that could be studied as part of future work.

In DPI, when the excitation and sample requirements discussed in Section \ref{DPIConvergenceAnalysis} are satisfied, the feedback gain $\hat{K}_i$ obtained from the Algorithm \ref{Algo4} provably stabilizes the true plant (this comes from standard true model-based PI results \cite{1099755}). Conversely, when the requirements on the convergence are not met, there are no closed-loop stability guarantees. 

\begin{table}
  \centering
  \begin{tabular}{|c|c|c|}
    \hline
    ~ & DPI & IPI \\\hline
    Assumptions & ~ & Assumption \ref{assumption 1},\ref{assumption 2} \\\hline
    Minimum episode length & $\tau_{\mathrm{DPI}}\geq\max\{n_x(n_x+1),\frac{n_u(n_u+1)}{2}+n_un_x\}$ & $\tau_{\mathrm{IPI}}\geq1$ \\ \hline
    \makecell[c]{Required persistency condition} & \makecell[c]{Local persistency,\\ $\phi$ with $N=M=\frac{\tau_\mathrm{DPI}}{2}$, \\$\Gamma$ with $N=M=\tau_\mathrm{DPI}$ }& \makecell[c]{Local persistency,\\ $d$ with any $N\geq1$, $M\geq1$}\\ \hline 
    \makecell[c]{Guarantee on \\persistency condition} & Not available a-priori & \makecell[c]{Achievable by selecting \\appropriate $e_t$}\\ \hline     
    Convergence rate & \eqref{convergenceDPI}&\eqref{convergenceIPI}\\ \hline
    Selection of excitation&  $u_t=\hat{K}_ix_t+\eta_t$ (on-policy) & Flexible (on/off-policy)\\ \hline    
  \end{tabular}
  \caption{Comparison between DPI and IPI.}\label{table}
\end{table}
\section{Simulation Results}\label{sec:simulation}
In this section, we present simulation results\footnote{The Matlab codes used to generate these results are accessible from the repository: \href{https://gitlab.com/col-tasas/datadrivenpolicyiteration-ipi-dpi}{https://gitlab.com/col-tasas/datadrivenpolicyiteration-ipi-dpi}} to illustrate some of the properties of IPI and DPI methods discussed in the previous sections.
\subsection{Comparison of DPI and IPI}
We consider the following system which was already used in prior studies \cite{1111,9691800,articlesimulation,AbbasiYadkori2018ModelFreeLQ}: 
\begin{equation}\label{LTIsimulation}
  x_{t+1}=\underbrace{\left[\begin{array}{ccc}
            1.01 & 0.01 & 0 \\
            0.01 & 1.01 & 0.01 \\
            0 & 0.01 & 1.01 
          \end{array}\right]}_A x_t+\underbrace{\left[\begin{array}{ccc}
            1 & 0 & 0 \\
            0 & 1 & 0 \\
            0 & 0 & 1 
          \end{array}\right]}_B u_t.
\end{equation}
The weighting matrices $Q$ and $R$ are set to $0.001I_3$ and $I_3$, respectively. The estimated $A$ and $B$ are initialized as zero matrices $0_3$ and $H_0$ for RLS is set to $0.01 I_6$. The initial stabilizing policy gain $K_1$ is set to $diag\{-1.5,-1.0,-0.5\}$ which stabilizes (\ref{LTIsimulation}) and also the initial estimate $\hat{A}_0$ and $\hat{B}_0$. 
The choice of dithering signals in the two algorithms is done as follows. In IPI 
$$e_t=\mathcal{N}(0, I_3),$$ whereas in DPI 
$$\eta_t =\left\{\begin{aligned}
   &\mathcal{N}(0, I_3), & &\mathrm{t~is~odd},\\
    &-\eta_{t-1}, & &\mathrm{t~is~even}.
\end{aligned}\right. $$

The simulation results are depicted in Figure \ref{fig:simulation1}.\\

\begin{figure}[h]
    \centering
    \includegraphics[width=0.6\textwidth]{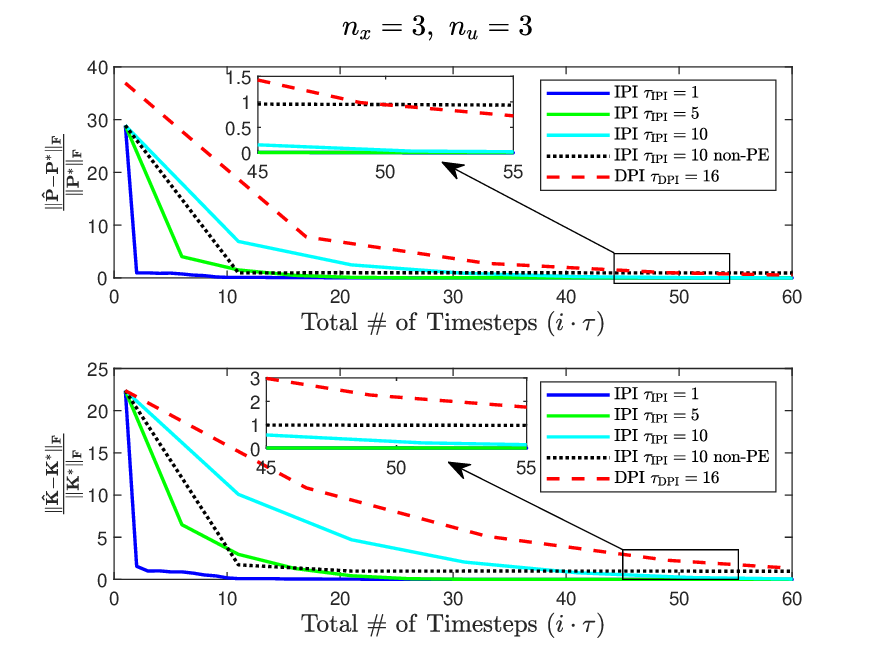}
    \caption{Comparison of DPI and IPI with different episode lengths $\tau_\mathrm{IPI}$ and minimum length $\tau_\mathrm{DPI}$ (the symbol $\tau$ in the $x$-axis label denotes the length of the episode in the respective approach).}
    \label{fig:simulation1}
\end{figure}

Figure \ref{fig:simulation1} illustrates the convergence of the quadratic kernel of the value function $\hat{P}_i$ and feedback gain $\hat{K}_i$ for different choices of episode length $\tau_\mathrm{IPI}$ in IPI. For this problem, the minimum number of samples per episode required by DPI is $15$, thus we always use $\tau_\mathrm{DPI}=16$, which includes redundant data to meet the requirement for the formulation of data pairs. For any $\tau_\mathrm{DPI}<16$, both $\hat{P}_i$ and $\hat{K}_i$ diverge. It can be observed that IPI has a faster convergence rate to the optimum when compared to DPI, in accordance with the discussion in Section \ref{IPI_DPI_rates}. The black dotted line shows the case when in IPI no dithering signal is used, i.e. $e_t=0$ and thus the data sequence is not locally persistent. Consequently, the estimation error of $\hat{\theta}_i$ is never zero, and $\hat{K}_i$ stays within a non-zero error from the optimal $K^*$, as expected from the theoretical analysis in Section \ref{IPIPI}. Overall the numerical experiments are in agreement with the analytical results discussed in the previous two Sections.\\

\subsection{Comparison with data-driven policy gradient}
We compare the IPI and DPI methods with a recently proposed method \cite{10383604} combining RLS with the model-based policy gradient (PG) method. We denote this method as RLS+PG. The system dynamics $(A,B)$ and the weighting matrices $Q,R$ are set according to the numerical experiment in \cite{10383604} as follows: 
\begin{align*}
    \begin{array}{cccc}
      A=\left[\begin{array}{ccc}
         -0.53  & 0.42  & -0.44\\
         0.42  &  -0.56 & -0.65\\
         -0.44  & -0.65 & 0.35
      \end{array}  \right], &  B=\left[\begin{array}{ccc}
         0.43  & -0.82  \\
         0.53  &  -0.78 \\
         0.26  & -0.40 
      \end{array}  \right], & Q=\left[\begin{array}{ccc}
         6.12  & 1.72  & 0.53\\
         1.72  &  6.86& 1.72\\
         0.53  & 1.72 & 5.73
      \end{array}  \right],  & R=\left[\begin{array}{cc}
         1.15  & -0.23  \\
         -0.23  &  3.62
      \end{array}  \right]
    \end{array}.
\end{align*}
The initial estimates $\hat{A}_0$,$\hat{B}_0$ and the matrix $H_0$, required in both IPI and RLS+PG, are set as zero matrices and $H_0=0.001 I_5$, respectively. The initial $\hat{K}_0$ is set to the optimal feedback gain for the LQR problem with $(A,B,100Q,R)$.
In RLS+PG, the system is excited as in IPI, i.e. with a feedback term plus dithering signal $e_t$ to ensure sufficiently informative data. 
The choice of dithering signals in the three algorithms is done as follows. In IPI and RLS+PG, 
$$e_t=\mathcal{N}(0, 3I_2), $$ whereas in DPI 
$$\eta_t =\left\{\begin{aligned}
   &\mathcal{N}(0, 3I_2), & &\mathrm{t~is~odd},\\
    &-\eta_{t-1}, & &\mathrm{t~is~even}.
\end{aligned}\right. $$
In the RLS part of the algorithm in \cite{10383604} the parameter $\lambda$ is set to $1$ , which recovers the classic RLS approach also used in this work. 
The effect of the step size $\gamma$ used in the PG part is investigated by simulating the values $0.01$ and $0.001$. This parameter determines the size of the gradient step and thus its value is critical to trade off convergence rate and robustness of the scheme. It is mentioned in \cite{10383604} that the selection of $\gamma$ is done empirically based on the initialization parameters. An analytical investigation of the selection of $\gamma$ is performed in \cite{pmlr-v80-fazel18a} for model-based PG. Determining an upper bound on the step size $\gamma$ which guarantees 
convergence of data-driven PG is an open problem.\\

\begin{figure}[h]
    \centering
    \includegraphics[width=0.6\textwidth]{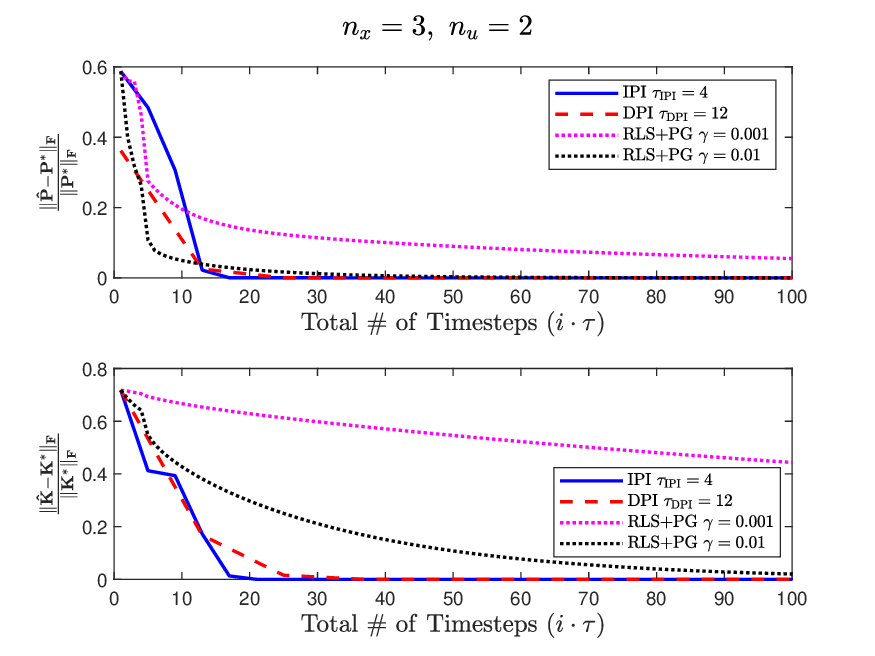}
    \caption{Comparison of IPI, DPI and RLS+PG}
    \label{fig:simulation2}
\end{figure}

Figure \ref{fig:simulation2} shows the convergence of the quadratic kernel $\hat{P}_i$ and gain $\hat{K}_i$ with DPI, IPI and the RLS+PG methods. It can be observed that both DPI and IPI exhibit faster convergence of $\hat{P}_i$ and $\hat{K}_i$ compared to the policy gradient methods. This observation is justifiable based on the fact that PI has an interpretation as Gauss-Newton method  \cite{9254115}, and thus it is expected to have better convergence rates. An analytic comparative study between these approaches in the data-based setting is an interesting future direction of research.
\section{Conclusion}\label{sec:conclusion}
In this work we studied indirect and direct approaches to data-driven LQR policy iteration with the goal to understand what are the main implications of having or not a model. We viewed the indirect approach as a feedback interconnection between an identification scheme and the policy iteration algorithm, both seen as instances of algorithmic dynamical systems responsible for the concurrent learning and control. 
We analysed this nonlinear closed-loop system without making any assumption on the level of excitation of the data and notably captured non-persistent data sequences as dynamic perturbations to the algorithm. 
For the direct approach, we extended a recently proposed method from the literature and precisely assessed its data and excitation requirements. 
In this problem setting, the availability of an intermediate model determines a lower number of required samples, the possibility to employ more flexible policies, and an increased robustness to lack of excitation in the data. 
The goal of the work was to isolate and understand the fundamental effects associated with the availability of a model, therefore no noise is considered in the analysis. This is a limitation of the current study and a comparison between indirect and direct approaches in the presence of process and measurement noise should be considered in future work. Moreover, an adaptive design of the excitation signal $e_t$ in IPI, which is guided by the current system theoretic analysis, is another important future direction.
\section*{Acknowledgements}
Bowen Song acknowledges the support of the International Max Planck Research School for Intelligent Systems (IMPRS-IS). Andrea Iannelli acknowledges the German Research Foundation (DFG) for support of this work under Germany’s Excellence Strategy - EXC 2075 – 390740016.

\bibliographystyle{unsrt}  
\bibliography{references}  

\newpage
\appendix
\section{Appendix 1}
\subsection{Lemma \ref{theorem21}} \label{App11}
\begin{lemma}\label{theorem21}
    For any sequence $\{d_t\}$, the upper bound of the estimation error $\Delta\theta^\mathrm{Upper}_t$ of the RLS is non-increasing with respect to $t$, i.e.
    \begin{equation}
        \Delta\theta_t^\mathrm{Upper} \geq \Delta\theta_{t+1}^\mathrm{Upper}, \forall t \in \mathbb{Z}_{+}
    \end{equation}
\end{lemma}
From \eqref{ConvergenceSystemid9}, we know that:
\begin{equation}\label{convergenceSystemid6}
   H_0 \preceq...\preceq H_t \preceq ...\preceq H_{\infty}.
\end{equation}
We then use the following result, which is proven in Appendix \ref{App:prooflemma1}
\begin{lemma}\label{lemma1}
  Given two positive definite matrix $A,B$, if $A\succeq B$, then $\lVert A^{-1} \rVert_F\leq\lVert B^{-1} \rVert_F$. 
\end{lemma}
Using \eqref{convergenceSystemid6} and Lemma \ref{lemma1}, the following holds for every $t\in \mathbb{Z}_{+}$: 
\begin{equation}\label{Proof1}
    \Delta\theta_t^\mathrm{Upper}=\lVert \Delta\theta_0H_0\rVert_F\lVert H_{t}^{-1} \rVert_F\geq \lVert \Delta\theta_0H_0\rVert_F\lVert H_{t+1}^{-1} \rVert_F=\Delta\theta_{t+1}^\mathrm{Upper}.
\end{equation}
Then we know that $\Delta\theta_t^\mathrm{Upper}$ is non-increasing with respect to time index $t$ irrespective of any property of $H_t$, i.e.$\{d_t\}$. Lemma \ref{theorem21} is now proven.
\subsection{Proof of Theorem \ref{theorem2}}\label{App1}
Now we focus on the convergence of the estimation error introduced by recursive least squares under the assumption of a locally persistent sequence defined by $M,N,\alpha$. When the data sequence $\{d_t\}$ is locally persistent, \eqref{ConvergenceSystemid8} holds. Then we can get:
\begin{equation}\label{Convergence7}
   \begin{split}
       \lVert \Delta\theta_t \rVert_F&\leq\lVert \Delta\theta_0H_0\rVert_F\lVert H_{t}^{-1} \rVert_F=\Delta\theta_t^\mathrm{Upper}\\
        &\leq \lVert \Delta\theta_0H_0\rVert_F\left\lVert \left(H_{0} +\sum\limits_{j=0}^{\left\lfloor\frac{t}{M\left\lceil\frac{N}{M}\right\rceil}\right\rfloor}\sum\limits_{k=0}^{N-1}d_{k+j(M\left\lceil\frac{N}{M}\right\rceil)}d_{k+j(M\left\lceil\frac{N}{M}\right\rceil)}^\top\right)^{-1}\right\rVert_F\\ 
        &\leq \lVert \Delta\theta_0H_0\rVert_F\left\lVert \left(\left(\left\lfloor\frac{t}{M\left\lceil\frac{N}{M}\right\rceil}\right\rfloor+1\right)\alpha I_{n_x+n_u}\right)^{-1}\right\rVert_F\\   
        &\leq \frac{\sqrt{n_x+n_u}\lVert \Delta\theta_0H_0\rVert_F}{\sqrt{\left(\left\lfloor\frac{t}{M\left\lceil\frac{N}{M}\right\rceil}\right\rfloor+1\right)\alpha}}.    
   \end{split} 
\end{equation}
The second inequality comes from \eqref{ConvergenceSystemid9} and the positive semidefinitiness of $d_t d_t^\top$. The third inequalities comes from \eqref{ConvergenceSystemid8}.\\
From \eqref{Convergence7} we can conclude that the convergence of RLS under the assumption of a locally persistent sequence $\{d_t\}$ is $\lVert \Delta\theta_t \rVert_F \sim \mathcal{O}(\frac{1}{\sqrt{\left\lfloor\frac{t}{M\left\lceil\frac{N}{M}\right\rceil}\right\rfloor}})$
\subsection{Proof of Lemma \ref{lemma1}}\label{App:prooflemma1}
For every positive definite matrix $A$, $\lVert A \rVert_F^2=\mathop{trace}(A^2)$\cite[Eq.541]{Petersen2008}. Because $A\succeq B\succ\mathbf{0}$, then $\lambda_i(A)\geq\lambda_i(B)>0$, $\lambda_i(B^{-1})\geq\lambda_i(A^{-1})>0$. We get:
\begin{equation}\label{Prooflemma1}
  \lVert A^{-1} \rVert_F^2=\mathop{trace}\left((A^{-1})^2\right)=\sum_{i}\limits\lambda_i(A^{-1})^2\leq\sum_{i}\limits\lambda_i(B^{-1})^2=\mathop{trace}\left((B^{-1})^2\right)=\lVert B^{-1} \rVert_F^2.
\end{equation} 
That is, it holds $\lVert A^{-1} \rVert_F\leq\lVert B^{-1} \rVert_F$.\\

\section{Appendix 2}
\subsection{Proof of Theorem \ref{theorem3}}\label{ISSRLS}
From \eqref{Proof1} and with $H_0=aI$, we can get:
\begin{equation}\label{RLS1}
  \begin{split}
      \lVert  \Delta\theta_i\rVert_F&\leq \Delta\theta_i^\mathrm{Upper}=\lVert \Delta\theta_0H_0\rVert_F\lVert H_{i}^{-1} \rVert_F=a\lVert \Delta\theta_0\rVert_FH_{i}^{-1}=a\lVert \Delta\theta_0\rVert_F\sqrt{\mathop{trace}((H_{i}^{-1})^2)}\\
      &=a\lVert \Delta\theta_0\rVert_F\sqrt{\sum\limits_{j=1}^{n_x+n_u}\lambda_j^2(H_{i}^{-1})}=a\lVert \Delta\theta_0\rVert_F\sqrt{\sum\limits_{j=1}^{n_x+n_u}\frac{1}{\lambda_j^2(H_{i})}}\leq a\lVert \Delta\theta_0\rVert_F\sum\limits_{j=1}^{n_x+n_u}\frac{1}{\lambda_k(H_{i})}\\
      &=a\lVert \Delta\theta_0\rVert_F\sum\limits_{j=1}^{n_x+n_u}\frac{1}{\lambda_j\left(H_0+\sum\limits_{k=1}^{i}D_{k}\right)}=a\lVert \Delta\theta_0\rVert_F\sum\limits_{j=1}^{n_x+n_u}\frac{1}{a+\lambda_j\left(\sum\limits_{k=1}^{i}D_{k}\right)}.
  \end{split}
\end{equation}
The inequality in \eqref{RLS1} comes from the inequality of arithmetic and the positive eigenvalues of the positive definite matrix.
\begin{equation}\label{RLS2}
  \begin{split}
      \lVert  \Delta\theta_i\rVert_F&\leq \Delta\theta_i^\mathrm{Upper}\leq a\lVert \Delta\theta_0\rVert_F\sum\limits_{j=1}^{n_x+n_u}\frac{1}{a+\lambda_j\left(\sum\limits_{k=1}^{i}D_{k}\right)} \\
      &= a\lVert \Delta\theta_0\rVert_F\sum\limits_{j=1}^{n_x+n_u}\frac{1}{a+\lfloor \frac{i}{N_{\max}}\rfloor{\alpha_{\min}}-\lfloor \frac{i}{N_{\max}}\rfloor{\alpha_{\min}} +\lambda_j\left(\sum\limits_{k=1}^{i}D_{k}\right)}\\
      &= \frac{a(n_x+n_u)\lVert \Delta\theta_0\rVert_F}{a+\lfloor \frac{i}{N_{\max}}\rfloor{\alpha_{\min}}}+\sum\limits_{j=1}^{n_x+n_u}\underbrace{\frac{a\lVert \Delta\theta_0\rVert_F\left(\lfloor \frac{i}{N_{\max}}\rfloor{\alpha_{\min}} -\lambda_j\left(\sum\limits_{k=1}^{i}D_{k}\right)\right)}{\left(a+\lambda_j\left(\sum\limits_{k=1}^{i}D_{k}\right)\right)\left( a+\lfloor \frac{i}{N_{\max}}\rfloor{\alpha_{\min}}\right)}}_{\mathcal{E}(j)}.
  \end{split}
\end{equation}
We analyze $\mathcal{E}(j)$ by focusing on the eigenvalues $\lambda_j\left(\sum\limits_{k=1}^{i}D_{k}\right)$. Because $\lambda_j\left(\sum\limits_{k=1}^{i}D_{k}\right) \geq 0$ for all $j$, the denominator of $\mathcal{E}(j)$ is  positive for every $j$. According to Definition \ref{def4}, $j^\mathrm{non}_i$ of the total $n_u+n_x$ eigenvalues satisfy the inequality $\lfloor\frac{i}{N_{\max}}\rfloor\alpha_{\min}-\lambda_j\left(\sum\limits_{k=1}^{i}D_{k}\right) >0 $. This implies that this part of the numerator of $\mathcal{E}(j)$ is upper-bounded by $a\lVert \Delta\theta_0\rVert_F\lfloor\frac{i}{N_{\max}}\rfloor\alpha_{\min}$. The remaining part is instead less than or equal to $0$. We can further simplify \eqref{RLS2} as: 
\begin{equation}\label{RLS34}
  \begin{split}
      \lVert  \Delta\theta_i\rVert_F&\leq  \Delta\theta_i^\mathrm{Upper}\leq \frac{a(n_x+n_u)\lVert \Delta\theta_0\rVert_F}{a+\lfloor \frac{i}{N_{\max}}\rfloor{\alpha_{\min}}}+\frac{j^\mathrm{non}_i a\lVert \Delta\theta_0\rVert_F\lfloor \frac{i}{N_{\max}}\rfloor{\alpha_{\min}}}{\left(a+\lambda_j\left(\sum\limits_{k=1}^{i}D_{k}\right)\right)\left( a+\lfloor \frac{i}{N_{\max}}\rfloor{\alpha_{\min}}\right)}
  \end{split}
\end{equation}
If $i \geq N_{\max}$, then from \eqref{RLS34} we can get:
\begin{equation}\label{RLS35}
   \begin{split}
      \lVert  \Delta\theta_i\rVert_F &\leq \Delta\theta_i^\mathrm{Upper}\leq \frac{a(n_x+n_u)\lVert \Delta\theta_0\rVert_F}{a+\lfloor \frac{i}{N_{\max}}\rfloor{\alpha_{\min}}}+\frac{j^\mathrm{non}_ia\lVert \Delta\theta_0\rVert_F\lfloor\frac{i}{N_{\max}}\rfloor\alpha_{\min} }{a\left(\lfloor \frac{i}{N_{\max}}\rfloor{\alpha_{\min}}\right)}\\
      &= a(n_x+n_u)\frac{\lVert \Delta\theta_0\rVert_F}{a+\lfloor \frac{i}{N_{\max}}\rfloor{\alpha_{\min}}}+\lVert \Delta\theta_0\rVert_Fj^\mathrm{non}_i
   \end{split}
\end{equation}
The first inequality is obtained by leveraging positivity of the eigenvalues in the denominator of the second term in (\ref{RLS34}). 
For the case $0\leq i \leq N_{\max}$, 
we know that the inequality \eqref{RLS35}, which can be simplified as $\lVert  \Delta\theta_i\rVert_F\leq \Delta\theta_i^\mathrm{Upper}\leq(n_x+n_u)\lVert \Delta\theta_0\rVert_F$, still holds. The proof of Theorem \ref{theorem3} is concluded by using the fact that $j^\mathrm{non}_i\leq\lVert j^\mathrm{non} \rVert_\infty, \forall i\in\mathbb{Z}_+$.
\section{Appendix 3}\label{convergenceNominal}
\subsection{Proof of Theorem \ref{theorem4}}\label{Prooftheorem6}
Before proving Theorem \ref{theorem4}, we introduce the following lemma:
\begin{lemma}\label{lemma2}
  Given a stabilizable system dynamic $\left(\tilde{A},\tilde{B}\right)$, and $\tilde{P}^*_{\left(\tilde{A},\tilde{B}\right)}$ is the solution of the DARE \eqref{DARE} formulated by the $\left(\tilde{A},\tilde{B}\right)$, 
  \begin{enumerate}
    \item For any $\tilde{P}\succeq P^*_{\left(\tilde{A},\tilde{B}\right)}$, the feedback gain $\tilde{K}=-\left(R+\tilde{B}^\top \tilde{P}\tilde{B}\right)^{-1}\tilde{B}^\top \tilde{P}\tilde{A}$ stabilizes the system, i.e. $\lVert \tilde{A}+\tilde{B}\tilde{K}\rVert_2 <1 $,   
    \item Defining $\tilde{\mathcal{A}}\left(\tilde{P}\right):=\left[I_{n_x}\otimes I_{n_x}-\left(\tilde{A}^\top +\tilde{B}^\top \tilde{K}^\top  \right)\otimes \left(\tilde{A}^\top +\tilde{B}^\top \tilde{K}^\top \right)\right]$, then $\lVert\tilde{\mathcal{A}}^{-1}\left(\tilde{P}_i\right)\rVert_F \leq \frac{\sqrt{n_x}}{1-\lVert \tilde{A}+\tilde{B}\tilde{K} \rVert^2_2}$,
    \item For any initialized $\tilde{P}_0\succeq P^*_{\left(\tilde{A},\tilde{B}\right)}$ and $\tilde{P}_{i+1}:=\mathcal{L}^{-1}_{\left(\tilde{A},\tilde{B},\tilde{P}_i\right)}\left(Q+\tilde{\alpha}^\top\left(\tilde{P}_i\right)\tilde{\beta}^{-1}\left(\tilde{P}_i\right)R\tilde{\beta}^{-1}\left(\tilde{P}_i\right)\tilde{\alpha}\left(\tilde{P}_i\right)\right)$ where $\tilde{\alpha}\left(\tilde{P}_i\right):=\tilde{B}^\top \tilde{P}_i\tilde{A},\tilde{\beta}\left(\tilde{P}_i\right):=R+\tilde{B}^\top \tilde{P}_i\tilde{B}$, there exist a constant $\tilde{\epsilon}\in [0,1)$, such that $\frac{\sqrt{n_x}}{1-\tilde{\epsilon}^2}\geq \lVert\tilde{\mathcal{A}}\left(\tilde{P}_i\right)\rVert_F$, for all $i\in\mathbb{Z}_+$.  
  \end{enumerate}
\end{lemma}
The proof of Lemma \ref{lemma2} is given in Section \ref{ProofLemma2}. We emphasize that Lemma \ref{lemma2} holds for any stabilizable system dynamic $\left(\tilde{A},\tilde{B}\right)$, for instance, the true system $(A,B)$ or a stabilizable estimate $\left(\hat{A}_i,\hat{B}_i\right)$.

Now we start to prove Theorem \ref{theorem4}. From item 1 in Lemma \ref{lemma2}, we know that $K_1$ stabilizes the system $(A,B)$ and thus Theorem \ref{theorem1} holds for the following analysis. By leveraging Theorem \ref{theorem1}, if $K_1$ stabilizes the system, then $ \forall i \in\mathbb{Z}_+, K_i$ also stabilizes the system. Thus, all the eigenvalues of $A+BK_{i}$ are inside the unit circle, i.e. $\lVert A+BK_{i}\rVert_2 <1 ~\forall i \in\mathbb{Z}_+$. Consequently, $\lVert \left(A^\top - \alpha^\top (P_i)\beta(P_i)^{-1}B^\top \right)\otimes \left(A^\top -\alpha^\top (P_i)\beta(P_i)^{-1}B^\top \right) \rVert_2<1$, as $A^\top - \alpha^\top (P_i)\beta(P_i)^{-1}B^\top =(A+BK_{i+1})^\top $. Then we can infer that the eigenvalues of $\mathcal{A}(P_i)$ are $1-\lambda\left(\left(A^\top - \alpha^\top (P_i)\beta(P_i)^{-1}B^\top \right)\otimes \left(A^\top -\alpha^\top (P_i)\beta(P_i)^{-1}B^\top \right)\right)$, all of which are greater than $0$. Thus, the matrix $\mathcal{A}(P_i),~\forall i=\mathbb{Z}_+$ is always invertible if $P_0\succeq P^*$. \\
Due to item 1 in Theorem \ref{theorem1}, we can conclude for any $P_i\succeq P^*$ that:
\begin{equation}\label{Convergence8}
  \lVert P_{i+1}-P^* \rVert_F \leq \lVert P_i-P^* \rVert_F.
\end{equation}
For any $P_i \neq P^*$, $P_i \neq P_{i+1}$ for all $i\in\mathbb{Z}_+$ because $P^*$ is the only equilibrium point of DARE \eqref{DARE}, we can conclude there always exists $\xi_i \in (0,1)$ such that:
\begin{equation}\label{Conevergence9}
   \lVert P_{i+1}-P^* \rVert_F \leq \xi_i \lVert P_i-P^* \rVert_F.
\end{equation}
For $P_i=P^*$, we can also find $\xi_i \in (0,1)$, such that \eqref{Conevergence9} holds. From item 4 in Theorem \ref{theorem1}, we know that $\lVert P_{i+1}-P^*\rVert_F \leq C\lVert P_{i}-P^*\rVert^2_F$, where $C$ is a constant. From item 3 in Theorem \ref{theorem1}, we know that there exists an index $i_c\geq 1$, such that $\lVert P_{i_c}-P^*\rVert_F <\frac{1}{C}$. Then we can always find a constant $c=\max \{\xi_1, \xi_2,...,\xi_{i_c},C\lVert P_{i_1}-P^*\rVert_F\}\in (0,1)$, such that $\lVert P_{i+1}-P^*\rVert_F \leq c\lVert P_{i}-P^*\rVert_F$ for all $i\in\mathbb{Z}_+$ and $P_0 \succeq P^*$.

\subsection{Proof of Lemma \ref{lemma2}}\label{ProofLemma2}
First, we proof item 1. We introduce the Lyapunov function $V(x_t):=x_t^\top \tilde{P}x_t$. Because of the stabilizability of $\left(\tilde{A},\tilde{B}\right)$, there exists a feedback gain $\tilde{K}^*$ that stabilizes the system and also minimizes the cost in \eqref{Cost}. From \eqref{costPI}, $\tilde{K}$ is the minimum of $J_\mathrm{PI}(K)$, then we have:
\begin{equation}\label{Lemma22}
  \begin{split}
     &x_t^\top Qx_t+\left(\tilde{K}x_t\right)^\top R\left(\tilde{K}x_t\right)+\left(\tilde{A}x_t+\tilde{B}\tilde{K}x_t\right)^\top \tilde{P}\left(\tilde{A}x_t+\tilde{B}\tilde{K}x_t\right)\\
       &\leq x_t^\top Qx_t+\left(\tilde{K}^*x_t\right)^\top R\left(\tilde{K}^*x_t\right)+\left(\tilde{A}x_t+\tilde{B}\tilde{K}^*x_t\right)^\top \tilde{P}\left(\tilde{A}x_t+\tilde{B}\tilde{K}^*x_t\right) \\
       & = x_t^\top P^*_{\left(\tilde{A},\tilde{B}\right)}x_t+\left(\tilde{A}x_t+\tilde{B}\tilde{K}^*x_t\right)^\top \left(\tilde{P}-P^*_{\left(\tilde{A},\tilde{B}\right)}\right)\left(\tilde{A}x_t+\tilde{B}\tilde{K}^*x_t\right)\\
       & = x_t^\top\tilde{P}x_t - \left(x_t^\top\left(\tilde{P}-P^*_{\left(\tilde{A},\tilde{B}\right)}\right)x_t-x_t^\top \left(\tilde{A}+\tilde{B}\tilde{K}^*\right)^\top \left(\tilde{P}-P^*_{\left(\tilde{A},\tilde{B}\right)}\right)\left(\tilde{A}+\tilde{B}\tilde{K}^*\right)x_t\right)
  \end{split}
\end{equation}
Subsequently, it follows that:
\begin{equation}\label{Lemma23}
  \begin{split}
     V(x_{t+1})-V(x_t) &= \left(\tilde{A}x_t+\tilde{B}\tilde{K}x_t\right)^\top \tilde{P}\left(\tilde{A}x_t+\tilde{B}\tilde{K}x_t\right)-x_t^\top \tilde{P}x_t\\
       &\leq -x_t^\top Qx_t-\left(\tilde{K}x_t\right)^\top R\left(\tilde{K}x_t\right)\\
       &-\left(x_t^\top \left(\tilde{P}-P^*_{\left(\tilde{A},\tilde{B}\right)}\right)x_t-x_t^\top \left(\tilde{A}+\tilde{B}\tilde{K}^*\right)^\top \left(\tilde{P}-P^*_{\left(\tilde{A},\tilde{B}\right)}\right)\left(\tilde{A}+\tilde{B}\tilde{K}^*\right)x_t\right)\\
       & \leq 0
  \end{split}
\end{equation}
The first inequality comes from \eqref{Lemma22} and the second inequality is because all eigenvalues of $\tilde{A}+\tilde{B}\tilde{K}^*$ are inside the unit circle. From \cite[Lemma 2.2]{articlematrix}, we know that $\left(\tilde{P}-P^*_{\left(\tilde{A},\tilde{B}\right)}\right)\succeq\left(\tilde{A}+\tilde{B}\tilde{K}^*\right)^\top \left(\tilde{P}-P^*_{\left(\tilde{A},\tilde{B}\right)}\right)\left(\tilde{A}+\tilde{B}\tilde{K}^*\right)$. From \eqref{Lemma23}, we can conclude that, $\tilde{K}$ stabilizes the system $\left(\tilde{A},\tilde{B}\right)$.\\
Now, we proceed with the proof of item 2. As established in item 1, all eigenvalues of the matrix $\left(\tilde{A}^\top+\tilde{B}^\top \tilde{K}^\top \right)\otimes \left(\tilde{A}^\top+\tilde{B}^\top \tilde{K}^\top \right)$ are inside of the unit circle with \cite[Eq.517]{Petersen2008}. Therefore, we can derive:
\begin{equation}\label{Lemma21}
  \begin{split}
     \lVert \tilde{\mathcal{A}}^{-1}\left(\tilde{P}\right)\rVert_F &\leq \sqrt{n_x}\lVert \tilde{\mathcal{A}}^{-1}\left(\tilde{P}\right)\rVert_2=\sqrt{n_x}\lVert \sum\limits_{k=0}^{\infty}[\left(\tilde{A}^\top +\tilde{B}^\top \tilde{K}^\top \right)\otimes \left(\tilde{A}^\top +\tilde{B}^\top \tilde{K}^\top \right)]^k\rVert_2\\
       &\leq \sqrt{n_x} \sum\limits_{k=0}^{\infty} \lVert [\left(\tilde{A}^\top +\tilde{B}^\top \tilde{K}^\top \right)\otimes \left(\tilde{A}^\top +\tilde{B}^\top \tilde{K}^\top \right)]^k\rVert_2\\
       &\leq  \sqrt{n_x}\sum\limits_{k=0}^{\infty} \lVert \tilde{A}^\top +\tilde{B}^\top \tilde{K}^\top  \rVert_2^{2k}=\sqrt{n_x}\sum\limits_{k=0}^{\infty} \lVert \tilde{A}^\top +\tilde{B}^\top \tilde{K}^\top  \rVert_2^{2k}\\
       &\leq \frac{\sqrt{n_x}}{1-\lVert \tilde{A}+\tilde{B}\tilde{K} \rVert^2_2}.
  \end{split}
\end{equation}
The first inequality follows from the definitions of induced 2-norm and Frobenius norm. The first equality comes from \cite[Eq. 186]{Petersen2008}. The second and third inequalities come from the sub-additive and sub-multiplicative of matrix norms.This shows Item 2 in Lemma \ref{lemma2}. By combining item 1 in Lemma \ref{lemma2} with Theorem \ref{theorem1}, it follows that $\hat{K}_1$ computed from $\hat{P}_0$ stabilizes the system and $\tilde{P}_i$ converges to $P^*_{\left(\tilde{A},\tilde{B}\right)}$, forming a convergent sequence $\{\tilde{P}_i\}$. Then we can always find $\tilde{\epsilon}:=\sup\{\lVert\tilde{A}+\tilde{B}\tilde{K}_i\rVert_2\}\in[0,1)$, such that $\frac{\sqrt{n_x}}{1-\tilde{\epsilon}^2}\geq \lVert\tilde{\mathcal{A}}\left(\tilde{P_i}\right)\rVert_F$, for all $i\in\mathbb{Z}_+$. Consequently, item 3 in Lemma \ref{lemma2} is also proven.
\section{Appendix 4}
\subsection{Proof of Theorem \ref{theorem5}}\label{prooftheorem5}
The pair $\left(A,B\right)$ is stabilizable, implying the existence of a gain $K_\mathrm{st}$ such that $\lVert A+BK_\mathrm{st}\rVert_2<1$. Let $L:=A+BK_\mathrm{st}$. According to Chapter 6 in \cite{meyer2023matrix}, the eigenvalues of the matrix $L\in \mathbb{R}^{n_x \times n_x}$ vary continuously with the entries of $L$. Therefore, there exists a $\Delta L$ such that $\lVert L+\Delta L \rVert_2<1$. Since $\lVert L\rVert_2+\lVert \Delta L \rVert_2\geq\lVert L+\Delta L \rVert_2$, it follows that if $\lVert \Delta L \rVert_2< 1-\lVert L\rVert_2$, then $\lVert L+\Delta L \rVert_2<1$. If $\lVert \Delta L\rVert_2\geq\lVert \Delta A_i+ \Delta B_i K_\mathrm{st}\rVert_2$, where $\Delta A_i:=\hat{A}_i -A$, $\Delta B_i:=\hat{B}_i -B$ and $\left(\hat{A}_i,\hat{B}_i\right)$ are the estimates from RLS, we can conclude $\lVert A+BK_\mathrm{st}+ \Delta A_i+ \Delta B_i K_\mathrm{st}\rVert_2=\lVert \hat{A}_i+\hat{B}_i K_\mathrm{st}\rVert_2<1$. Because of $\Delta\theta_i=\left[\Delta A_i,\Delta B_i\right]$, we know that:
\begin{equation}\label{T51}
   \lVert \Delta A_i \rVert_2 ~\mathrm{or} ~\lVert \Delta B_i \rVert_2 \leq \lVert \Delta\theta_i \rVert_2.
\end{equation}
If $\lVert \Delta\theta_i \rVert_2 \leq \frac{\lVert \Delta L \rVert_2}{1+\lVert K_\mathrm{st} \rVert_2}$, then we have:
\begin{equation}\label{T52}
  \begin{split}
     \lVert \Delta A_i+ \Delta B_i K_\mathrm{st}\rVert_2 &\leq \lVert \Delta A_i\rVert_2 + \lVert \Delta B_i \rVert_2 \lVert K_\mathrm{st}\rVert_2 \\
       & \leq \lVert  \Delta\theta_i \rVert_2 + \lVert  \Delta\theta_i  \rVert_2 \lVert K_\mathrm{st}\rVert_2\leq\lVert \Delta L \rVert_2.
  \end{split}
\end{equation}
From \eqref{T52}, $K_\mathrm{st}$ stabilizes the system $\left(\hat{A}_i,\hat{B}_i\right)$, and thus, $\left(\hat{A}_i,\hat{B}_i\right)$ is stabilizable. By Theorem \ref{theorem2}, we know that when the sequence is locally persistent with respect to any $N,M, \alpha$, then $\lim \limits_{i\rightarrow\infty}\Delta\theta_i=0$. This means there  exists an index $i_\mathrm{st}$ such that for all $i\geq i_\mathrm{st}$, $\lVert \Delta\theta_i \rVert_2 \leq \frac{\lVert \Delta L \rVert_2}{1+\lVert K_\mathrm{st} \rVert_2}$, then the estimates $\left(\hat{A}_i,\hat{B}_i\right)$ are always stabilizable. The proof is complete.
\subsection{Proof of Theorem \ref{theorem6}}\label{App2}
First, we prove that the IPI algorithm formulated by \eqref{IPI2a}  \eqref{procedure} and \eqref{IPI2d} can be equivalently written as \eqref{IPI2}. For the recursive least squares state representation \eqref{IPI21c}, we replace the expressions of $H_i$ in \eqref{IPI2c} by \eqref{IPI2b}. For the policy iteration state representation \eqref{IPI21d}, we can replace the expression of $K_i$ in \eqref{IPI2a} by \eqref{IPI2d} with shifting index from $i$ to $i-1$. However, we must check if the inverse operation $\mathcal{L}^{-1}_{\left(\hat{A}_i,\hat{B}_i,P_i\right)}$ still exists. With Assumption \ref{assumption 1}, Assumption \ref{assumption 2} and item 1 in Lemma \ref{lemma2}, we know that $\hat{K}_i$ stabilizes the system $\left(\hat{A}_i,\hat{B}_i\right)$. Thus, the inverse operation $\mathcal{L}^{-1}_{\left(\hat{A}_i,\hat{B}_i,\hat{P}_i\right)}$ exists and $\hat{P}_{i+1}$ can be expressed as \eqref{IPI21d}.\\

Then, we proceed to proving the result in \eqref{ISSCA_TOT}. For \eqref{ISSCA2} we can simply use Theorem \eqref{theorem3}. Indeed this holds for any data sequence in input to the RLS scheme, hence also for the one generated by IPI. We can the focus just on \eqref{ISSCA}. 
At episode $i$, $\hat{P}_{i+1}$ is updated with the estimates $\left(\hat{A}_i,\hat{B}_i\right)$ and $\hat{P}_{i}$.  
$\hat{P}_{i+1}$ can be calculated by \eqref{IPI21d} as follows : 
\begin{equation}\label{ISS1}
\begin{split}
   \hat{P}_{i+1} &= \mathcal{L}^{-1}_{\left(\hat{A}_i,\hat{B}_i,\hat{P}_i\right)}\left(Q+\alpha_i^\top \left(\hat{P}_i\right)\beta_i^{-1}\left(\hat{P}_i\right)R\beta_i^{-1}\left(\hat{P}_i\right)\alpha_i\left(\hat{P}_i\right)\right).
\end{split}
\end{equation}
We rewrite \eqref{ISS1} as two terms, one term denotes the update of $\hat{P}_i$ with respect to \eqref{ISS3} formulated by the true system $\left(A,B\right)$, the other term denotes the difference between \eqref{ISS3} formulated by $\left(\hat{A}_i,\hat{B}_i\right)$ and $\left(A,B\right)$:
\begin{equation}\label{ISS111}
\begin{split}
   \hat{P}_{i+1}  &= \mathcal{L}^{-1}_{\left(A,B,\hat{P}_i\right)}\left(Q+\alpha^\top P\left(\hat{P}_i\right)\beta^{-1}\left(\hat{P}_i\right)R\beta^{-1}\left(\hat{P}_i\right)\alpha\left(\hat{P}_i\right)\right)+\varepsilon\left(\Delta A_i,\Delta B_i\right),
\end{split}
\end{equation}
where \begin{equation}\label{iss5}
  \begin{split}
      \varepsilon\left(\Delta A_i,\Delta B_i\right):=&\mathcal{L}^{-1}_{\left(\hat{A}_i,\hat{B}_i,\hat{P}_i\right)}\left(Q+\alpha_i^\top \left(\hat{P}_i\right)\beta_i^{-1}\left(\hat{P}_i\right)R\beta_i^{-1}\left(\hat{P}_i\right)\alpha_i\left(\hat{P}_i\right)\right)\\
      &-\mathcal{L}^{-1}_{\left(A,B,\hat{P}_i\right)}\left(Q+\alpha^\top \left(\hat{P}_i\right)\beta^{-1}\left(\hat{P}_i\right)R\beta^{-1}\left(\hat{P}_i\right)\alpha\left(\hat{P}_i\right)\right).
  \end{split} 
\end{equation}
Now, we aim to find the upper bound of $\lVert \varepsilon\left(\Delta A_i,\Delta B_i\right)\rVert_F$ in terms of $\lVert \Delta\theta_i \rVert_F $. The following lemma from \cite[Equation 19]{9444823} is used, when $X,X+\Delta X$ are invertible:
\begin{equation}\label{ISS5}
   \lVert X^{-1}Y-\left(X+\Delta X\right)^{-1}\left(Y+\Delta Y\right) \rVert_F \leq \lVert X^{-1}\rVert_F\left(\lVert \Delta Y\rVert_F+\lVert \left(X+\Delta X\right)^{-1}\rVert_F\lVert \left(Y+\Delta Y\right)\rVert_F\lVert \Delta X\rVert_F\right).
\end{equation}
With \eqref{convergence6} and \eqref{Proof1}, we can get:
\begin{equation}\label{ISS9}
   \lVert \Delta A_i \rVert_F ~\mathrm{or} ~\lVert \Delta B_i \rVert_F \leq \lVert \Delta\theta_i \rVert_F\leq \Delta\theta^\mathrm{Upper}_i \leq \Delta\theta^\mathrm{Upper}_0.
\end{equation}
We define the updated gain $\hat{K}_{i+1}$ as the gain computed with the current estimated system matrix $\left(\hat{A}_i,\hat{B}_i\right)$ and $\hat{P}_i$ and $\bar{K}_{i+1}$ as  the gain computed with $\left(A,B\right)$ and $\hat{P}_i$.
\begin{equation}\label{ISS6}
  \begin{split}
     \lVert \Delta K_{i+1} \rVert_F&=\lVert \hat{K}_{i+1}-\bar{K}_{i+1}\rVert_F \\
       &=\lVert \left(R+B^\top \hat{P}_iB\right)^{-1}B^\top \hat{P_i}A-\left(R+B^\top \hat{P}_iB+\Delta_{\beta_i}\right)^{-1}\left(B^\top \hat{P}_iA+\Delta_{\alpha_i}\right) \rVert_F\\
       &\leq \lVert \left(R+B^\top \hat{P}_iB\right)^{-1}\rVert_F\left(\lVert \Delta_{\alpha_i}\rVert_F+\lVert \left(R+\hat{B}_i^\top \hat{P}_i\hat{B}_i\right)^{-1}\rVert_F\lVert \hat{B}_i^\top \hat{P}_i\hat{A}_i\rVert_F\lVert \Delta_{\beta_i}\rVert_F\right),
  \end{split}
\end{equation}
where $\Delta_{\beta_i},\Delta_{\alpha_i}$ are given as
\begin{equation}\label{ISS7}
   \Delta_{\beta_i}=\Delta B_i^\top  \hat{P}_i \Delta B_i+ B^\top \hat{P}_i \Delta B_i+\Delta  B_i^\top  \hat{P}_i  B,
\end{equation}
\begin{equation}\label{ISS8}
  \Delta_{\alpha_i}=\Delta B_i^\top  \hat{P}_i \Delta A_i+ B^\top \hat{P}_i \Delta A_i+\Delta  B_i^\top \hat{P}_i  A. 
\end{equation}
The inequality in \eqref{ISS6} comes from \eqref{ISS5}. Given any stabilizable $\left(\hat{A}_i,\hat{B}_i\right)$, from Theorem \ref{theorem1} and Lemma \ref{lemma2}, we know that $\lVert \hat{P}_{i+1} \rVert_F$ is bounded by $\lVert \hat{P}_{i} \rVert_F$. Then we can get $\lVert \hat{P}_i \rVert_F \leq\lVert \hat{P}_0 \rVert_F$. Together with \eqref{ISS7} \eqref{ISS8} and \eqref{ISS9}, the norm of $\lVert \Delta_{\beta_i}\rVert_F,\lVert \Delta_{\alpha_i}\rVert_F$ are bounded by:
\begin{equation}\label{ISS10}
  \begin{split}
      \lVert \Delta_{\beta_i} \rVert_F&=\lVert\Delta B_i^\top  \hat{P}_i \Delta B_i+ B^\top \hat{P}_i \Delta B_i+\Delta  B_i^\top  \hat{P}_i  B\rVert_F\\
      &\leq\lVert \Delta\theta_i \rVert_F\underbrace{\lVert\hat{P}_0\rVert_F\left(\Delta\theta^\mathrm{Upper}_i+2\lVert B \rVert_F\right)}_{=:\delta_1}=\delta_1\lVert \Delta\theta_i \rVert_F,
  \end{split}  
\end{equation}
\begin{equation}\label{ISS11}
  \begin{split}
      \lVert\Delta_{\alpha_i}\rVert_F&=\lVert\Delta B_i^\top  \hat{P}_i \Delta A_i+ B^\top \hat{P}_i \Delta A_i+\Delta  B_i^\top  \hat{P}_i  A \rVert_F\\
      &\leq\lVert \Delta\theta_i \rVert_F \underbrace{\lVert \hat{P}_{0}\rVert_F\left(\Delta\theta^\mathrm{Upper}_i+\lVert B\rVert_F+\lVert A\rVert_F\right)}_{=:\delta_2}=\delta_2\lVert \Delta\theta_i \rVert_F.
  \end{split}
\end{equation}
The inequalities in \eqref{ISS10} and \eqref{ISS11} come from the sub-multiplicative property and item 1 in Theorem \ref{theorem1}. Then \eqref{ISS6} can be further simplified as:
\begin{equation}\label{ISS12}
  \begin{split}
    \lVert \hat{K}_{i+1}-\bar{K}_{i+1}\rVert_F &\leq \lVert \left(R+B^\top \hat{P}_iB\right)^{-1}\rVert_F\left(\lVert \Delta_{\alpha_i}\rVert_F+\lVert \left(R+\hat{B}_i^\top \hat{P}_i\hat{B}_i\right)^{-1}\rVert_F\lVert \hat{B}_i^\top \hat{P}_i\hat{A}_i\rVert_F\lVert \Delta_{\beta_i}\rVert_F\right) \\
    &\leq \lVert R^{-1}\rVert_F\left(\lVert \Delta_{\alpha_i}\rVert_F+\lVert R^{-1}\rVert_F\lVert \Delta B_i+B\rVert_F\lVert\hat{P}_i\rVert_F\lVert\Delta A_i+A\rVert_F\lVert \Delta_{\beta_i}\rVert_F\right) \\
     &\leq \delta_3\lVert \Delta \theta_i \rVert_F,
  \end{split}
\end{equation}
where $\delta_3:=\delta_{33}\left(\Delta\theta^\mathrm{Upper}_i\right)^3+\delta_{32}\left(\Delta\theta^\mathrm{Upper}_i\right)^2+\delta_{31}\Delta\theta^\mathrm{Upper}_i+\delta_{30}$, with $\delta_{33}:=\lVert R^{-1}\rVert_F^2\lVert \hat{P}_{0}\rVert_F^2$, \\
$\delta_{32}:=2\lVert R^{-1}\rVert_F^2\lVert \hat{P}_{0}\rVert_F^2\left(\lVert A\rVert_F+\lVert B\rVert_F\right)$, 
$\delta_{31}:=\lVert R^{-1}\rVert_F\lVert \hat{P}_{0}\rVert_F+\lVert R^{-1}\rVert_F^2\lVert \hat{P}_{0}\rVert_F^2\left(\left(\lVert A\rVert_F+\lVert B\rVert_F\right)^2+\lVert A\rVert_F\lVert B\rVert_F\right)$,\\
$\delta_{30}:=2\lVert R^{-1}\rVert_F\lVert \hat{P}_{0}\rVert_F\lVert B\rVert_F+\lVert R^{-1}\rVert_F^2\lVert \hat{P}_{0}\rVert_F^2\lVert A\rVert_F\lVert B\rVert_F\left(\lVert A\rVert_F+\lVert B\rVert_F\right)$.\\

The first inequality comes from the positive semidefiniteness of $B^\top \hat{P}_iB$ and $\hat{B}_i^\top \hat{P}_i\hat{B}_i$. The second inequality comes from \eqref{ISS9} \eqref{ISS10} and \eqref{ISS11}. In \eqref{rerelationPseq2}, $\mathcal{A}\left(\hat{P}_i\right)=\left[{I}_{n_x}\otimes{I}_{n_x}-\left(A^\top + \bar{K}_i^\top B^\top\right)\otimes \left(A^\top + \bar{K}_i^\top B^\top \right)\right]$. Similarly, we define $\mathcal{A}_i\left(\hat{P}_i\right):=\left[{I}_{n_x}\otimes{I}_{n_x}-\left(\hat{A}_i^\top + \hat{K}_i^\top \hat{B}_i^\top \right)\otimes \left(\hat{A}_i^\top + \hat{K}_i^\top \hat{B}_i^\top \right)\right]$ and $\Delta\mathcal{A}_i\left(\hat{P}_i\right):=\mathcal{A}_i\left(\hat{P}_i\right)-\mathcal{A}\left(\hat{P}_i\right)$, 
\begin{equation}\label{ISS13}
\begin{split}
     \lVert \Delta\mathcal{A}_i\left(\hat{P_i}\right)\rVert_F&=\lVert \mathcal{A}\left(\hat{P_i}\right)- \mathcal{A}_i\left(\hat{P_i}\right)\rVert_F \\
     &=\lVert \left(\hat{A}_i+ \hat{B}_i\tilde{K}_{i+1}\right)\otimes \left(\hat{A}_i+ \hat{B}_i\tilde{K}_{i+1}\right)- \left(A+ B\bar{K}_{i+1}\right)\otimes \left(A+ B\bar{K}_{i+1}\right)\rVert_F \\
     &=\lVert \left(A+\Delta A_i+ \left(B+\Delta B_i\right)\left(\bar{K}_{i+1}+\Delta{K}_{i+1}\right)\right)\otimes \left(\hat{A}_i+ \hat{B}_i\tilde{K}_{i+1}\right)-(\underbrace{A+ B\bar{K}_{i+1}}_{=:A_{K_{i+1}}}) \otimes (A_{K_{i+1}})\rVert_F \\
     &=\lVert (A_{K_{i+1}}+\underbrace{\Delta A_i+B\Delta{K}_{i+1}+ \Delta B_i(\bar{K}_{i+1}+\Delta{K}_{i+1}}_{=:\mathcal{C}_i}))\otimes (A_{K_{i+1}}+\mathcal{C}_i)- \left(A_{K_{i+1}}\right)\otimes \left(A_{K_{i+1}}\right) \rVert_F.
\end{split}
\end{equation}
 Because $\lVert \bar{K}_{i+1}\rVert_F\leq \lVert R^{-1}\rVert_F \lVert B\rVert_F\lVert \hat{P}_i\rVert_F\lVert A\rVert_F \leq \lVert R^{-1}\rVert_F \lVert B\rVert_F\lVert \hat{P}_0\rVert_F\lVert A\rVert_F$, then $\lVert \mathcal{C}_i \rVert_F \leq \delta_4 \lVert \Delta \theta_i \rVert_F$, where $\delta_4:=\delta_{44}\left(\Delta\theta^\mathrm{Upper}_i\right)^4+\delta_{43}\left(\Delta\theta^\mathrm{Upper}_i\right)^3+ \delta_{42}\left(\Delta\theta^\mathrm{Upper}_i\right)^2+\delta_{41}\Delta\theta^\mathrm{Upper}_i+\delta_{40}$, with $\delta_{44}:=\delta_{33}$, $\delta_{43}:=\delta_{32}+\lVert B\rVert_F\delta_{33}$, $\delta_{42}:=\delta_{31}+\lVert B\rVert_F\delta_{32}$, $\delta_{41}:=\delta_{30}+\lVert B\rVert_F\delta_{31}$, $\delta_{40}:=\lVert B\rVert_F\delta_{30}+\lVert R^{-1}\rVert_F\lVert A\rVert_F\lVert B\rVert_F \lVert\hat{P}_{0}\rVert_F+1$.
 
Consequently, \eqref{ISS13} can be further simplified as: 
\begin{equation}\label{ISS14}
\begin{split}
     \lVert \Delta\mathcal{A}_i\left(\hat{P_i}\right)\rVert_F &=\lVert\mathcal{C}_i\otimes\left(A_{K_{i+1}}\right)+\left(A_{K_{i+1}}\right)\otimes\mathcal{C}_i+\mathcal{C}_i\otimes\mathcal{C}_i\rVert_F\\
     &\leq 2\lVert\mathcal{C}_i\rVert_F\lVert A+B\bar{K}_{i+1} \rVert_F+\lVert\mathcal{C}_i\rVert_F^2\\
     &\leq \delta_5 \lVert \Delta \theta_i \rVert_F, 
\end{split}
\end{equation}
where $\delta_5:=\delta_{59}\left(\Delta\theta^\mathrm{Upper}_i\right)^9+\delta_{58}\left(\Delta\theta^\mathrm{Upper}_i\right)^8+\delta_{57}\left(\Delta\theta^\mathrm{Upper}_i\right)^7+\delta_{56}\left(\Delta\theta^\mathrm{Upper}_i\right)^6+\delta_{55}\left(\Delta\theta^\mathrm{Upper}_i\right)^5+ \delta_{54}\left(\Delta\theta^\mathrm{Upper}_i\right)^4+\delta_{53}\left(\Delta\theta^\mathrm{Upper}_i\right)^3+ \delta_{52}\left(\Delta\theta^\mathrm{Upper}_i\right)^2+\delta_{51}\Delta\theta^\mathrm{Upper}_i+\delta_{50}$, with $\delta_{59}:=\delta_{44}\delta_{44}$, $\delta_{58}:=2\delta_{44}\delta_{43}$, $\delta_{57}:=2\delta_{44}\delta_{42}+\delta_{43}\delta_{43}$, $\delta_{56}:=2\delta_{44}\delta_{41}+2\delta_{43}\delta_{42}$,$\delta_{55}:=2\delta_{44}\delta_{40}+2\delta_{43}\delta_{41}+\delta_{42}\delta_{42}$,
$\delta_{54}:= 2\delta_{43}\delta_{40}+2\delta_{42}\delta_{41}+2\left(\lVert A\rVert_F+\lVert R^{-1}\rVert_F\lVert A\rVert_F\lVert B\rVert_F^2 \lVert\hat{P}_{0}\rVert_F\right)\delta_{44}$, $\delta_{53}:=2\delta_{42}\delta_{40}+\delta_{41}\delta_{41}+2\left(\lVert A\rVert_F+\lVert R^{-1}\rVert_F\lVert A\rVert_F\lVert B\rVert_F^2 \lVert\hat{P}_{0}\rVert_F\right)\delta_{43}$, $\delta_{52}:=2\delta_{41}\delta_{40}+2\left(\lVert A\rVert_F+\lVert R^{-1}\rVert_F\lVert A\rVert_F\lVert B\rVert_F^2 \lVert\hat{P}_{0}\rVert_F\right)\delta_{42}$, $\delta_{51}:=\delta_{40}\delta_{40}+2\left(\lVert A\rVert_F+\lVert R^{-1}\rVert_F\lVert A\rVert_F\lVert B\rVert_F^2 \lVert\hat{P}_{0}\rVert_F\right)\delta_{41}$, $\delta_{50}:=2\left(\lVert A\rVert_F+\lVert R^{-1}\rVert_F\lVert A\rVert_F\lVert B\rVert_F^2 \lVert\hat{P}_{0}\rVert_F\right)\delta_{40}$.\\

The first inequality comes from $\lVert M \otimes N\rVert_F=\lVert M \rVert_F\lVert N\rVert_F$ and the second inequality comes from \eqref{ISS12}.
Defining $\tilde{K}_{i+1}^{\top }R\tilde{K}_{i+1}:=\bar{K}_{i+1}^{\top }R\bar{K}_{i+1}+\Delta_{KRK}$ and together with \eqref{ISS12}, we can also get:
\begin{equation}\label{ISS15}
  \begin{split}
     \lVert \Delta_{KRK} \rVert_F &= \lVert \left(\bar{K}_{i+1}+\Delta{K}_{i+1}\right)^{\top }R\left(\bar{K}_{i+1}+\Delta{K}_{i+1}\right)-\bar{K}_{i+1}^{\top }R\bar{K}_{i+1} \rVert_F \\
       &= \lVert \bar{K}_{i+1}^{\top } R \Delta{K}_{i+1}+\Delta{K}_{i+1}^{\top } R\bar{K}_{i+1}+ \Delta{K}_{i+1}^{\top } R \Delta{K}_{i+1}\rVert_F \\
       &\leq \delta_{6}\lVert \Delta \theta_i \rVert_F,
  \end{split}
\end{equation}
where $\delta_6:=\delta_{67} \left(\Delta\theta^\mathrm{Upper}_i\right)^7+\delta_{66} \left(\Delta\theta^\mathrm{Upper}_i\right)^6+\delta_{65}\left(\Delta\theta^\mathrm{Upper}_i\right)^5+\delta_{64}\left(\Delta\theta^\mathrm{Upper}_i\right)^4+\delta_{63}\left(\Delta\theta^\mathrm{Upper}_i\right)^3+ \delta_{62}\left(\Delta\theta^\mathrm{Upper}_i\right)^2+\delta_{61}\Delta\theta^\mathrm{Upper}_i+\delta_{60}$, with $\delta_{67}:=\lVert R\rVert_F\delta_{33}\delta_{33}$, $\delta_{66}:=2\lVert R\rVert_F\delta_{33}\delta_{32}$, $\delta_{65}:=2\lVert R\rVert_F\delta_{33}\delta_{31}+\lVert R\rVert_F\delta_{32}\delta_{32}$, $\delta_{64}:=2\lVert R\rVert_F\delta_{33}\delta_{30}+2\lVert R\rVert_F\delta_{32}\delta_{31}$, 
$\delta_{63}:=\lVert R\rVert_F\delta_{31}\delta_{31}+2\lVert B\rVert_F\lVert \hat{P}_0\rVert_F \lVert R\rVert_F \lVert R^{-1}\rVert_F\lVert A\rVert_F\delta_{33}$, $\delta_{62}:=2\lVert R\rVert_F\delta_{31}\delta_{30}+2\lVert B\rVert_F\lVert \hat{P}_0\rVert_F \lVert R\rVert_F \lVert R^{-1}\rVert_F\lVert A\rVert_F\delta_{32}$, 
$\delta_{61}:=\lVert R\rVert_F\delta_{30}\delta_{30}+2\lVert B\rVert_F\lVert \hat{P}_0\rVert_F \lVert R\rVert_F \lVert R^{-1}\rVert_F\lVert A\rVert_F\delta_{31}$, $\delta_{60}:=2\lVert B\rVert_F\lVert \hat{P}_0\rVert_F\lVert R\rVert_F \lVert R^{-1}\rVert_F \lVert A\rVert_F\delta_{30}$.\\

We plug \eqref{iss5} \eqref{ISS5} \eqref{ISS14} \eqref{ISS15} into \eqref{ISS5}, then we can obtain:
\begin{equation}\label{ISS4}
  \begin{split}
     \lVert \varepsilon\left(\Delta A_i,\Delta B_i\right) \rVert_F &= \lVert \mathcal{A}^{-1}\left(\hat{P_i}\right)\left(Q+\bar{K}_{i+1}^{\top }R\bar{K}_{i+1}\right)-\mathcal{A}^{-1}_i\left(\hat{P_i}\right)\left(Q+\tilde{K}_{i+1}^{\top }R\tilde{K}_{i+1}\right)\rVert_F\\
       &=\lVert \mathcal{A}^{-1}\left(\hat{P}_i\right)\left(Q+\bar{K}_{i+1}^{\top }R\bar{K}_{i+1}\right)-\left(\mathcal{A}\left(\hat{P}_i\right)+\Delta\mathcal{A}_i\left(\hat{P}_i\right)\right)^{-1}\left(Q+\bar{K}_{i+1}^{\top }R\bar{K}_{i+1}+\Delta_{KRK}\right)\rVert_F\\
       &\leq \lVert \mathcal{A}^{-1}\left(\hat{P}_i\right) \rVert_F \left(\delta_{6}+\delta_{5}\lVert \mathcal{A}^{-1}_i\left(\hat{P}_i\right) \rVert_F \lVert Q+\bar{K}_{i+1}^{\top }R\bar{K}_{i+1}\rVert_F \right)\lVert \Delta \theta_i \rVert_F.
  \end{split}
\end{equation}
From item 3 in Lemma \ref{lemma2}, we define $\epsilon := \sup\{\lVert A+B\hat{K}_i\rVert_2\}\in[0,1)$, such that $\frac{\sqrt{n_x}}{1-\epsilon^2}\geq \lVert\mathcal{A}^{-1}\left(\hat{P_i}\right)\rVert_F$, for all $i\in\mathbb{Z}_+$. From item 2 in Lemma \ref{lemma2}, we know that $\hat{K}_i$ stabilizes the system $\left(\hat{A}_i,\hat{B}_i\right)$, i.e. $\epsilon_i:=\lVert \hat{A}_i+\hat{B}_i\hat{K}_i \rVert_2\in[0,1)$, for all $i\in\mathbb{Z}_+$. Then we can always find a $\epsilon_{\max}:=\sup\{\epsilon_i\}$, such that $\lVert \mathcal{A}^{-1}_i\left(\hat{P_i}\right) \rVert_F\leq\frac{\sqrt{n_x}}{1-\epsilon_{\max}^2}$. Now we can further simplify \eqref{ISS4}:
\begin{equation}\label{ISS44}
  \begin{split}
     \lVert \varepsilon\left(\Delta A_i,\Delta B_i\right) \rVert_F &\leq \lVert \mathcal{A}^{-1}\left(\hat{P_i}\right) \rVert_F \left(\delta_{6}+\delta_{5}\lVert \mathcal{A}^{-1}_i\left(\hat{P_i}\right) \rVert_F \lVert Q+\bar{K}_{i+1}^{\top }R\bar{K}_{i+1}\rVert_F \right)\lVert \Delta \theta_i \rVert_F\\
       &\leq \underbrace{\frac{\sqrt{n_x}}{1-\epsilon^2} \left(\delta_{6}+\delta_{5}\frac{\sqrt{n_x}}{1-\epsilon_{\max}^2} \left(\lVert Q\rVert_F+\lVert\bar{K}_{i+1}\rVert_F \lVert R \rVert_F \lVert \bar{K}_{i+1}\rVert_F\right) \right)}_{=:\sigma_i} \lVert \Delta \theta_i \rVert_F\\
       &=\sigma_i \lVert \Delta \theta_i \rVert_F \leq \sigma_i \Delta\theta_i^\mathrm{Upper},
  \end{split}
\end{equation}
where $\sigma_i=\sum\limits_{k=0}^{9}\rho_k\left(\Delta\theta^\mathrm{Upper}_i\right)^k$ and the detailed expressions of $\rho_k>0$ are given as follows: \\
$\rho_9:=\frac{\sqrt{n_x}}{1-\epsilon^2}\frac{\sqrt{n_x}}{1-\epsilon_{max}^2}\left(\lVert Q\rVert_F+\lVert A\rVert_F^2\lVert B\rVert_F^2\lVert R\rVert_F\lVert R^{-1}\rVert_F^2\lVert \hat{P}_0\rVert_F^2\right)\delta_{59}$,\\
$\rho_8:=\frac{\sqrt{n_x}}{1-\epsilon^2}\frac{\sqrt{n_x}}{1-\epsilon_{max}^2}\left(\lVert Q\rVert_F+\lVert A\rVert_F^2\lVert B\rVert_F^2\lVert R\rVert_F\lVert R^{-1}\rVert_F^2\lVert \hat{P}_0\rVert_F^2\right)\delta_{58}$, \\
$\rho_7:=\frac{\sqrt{n_x}}{1-\epsilon^2}\delta_{67}+\frac{\sqrt{n_x}}{1-\epsilon^2}\frac{\sqrt{n_x}}{1-\epsilon_{max}^2}\left(\lVert Q\rVert_F+\lVert A\rVert_F^2\lVert B\rVert_F^2\lVert R\rVert_F \lVert R^{-1}\rVert_F^2\lVert\hat{P}_0\rVert_F^2\right)\delta_{57}$, \\
$\rho_6:=\frac{\sqrt{n_x}}{1-\epsilon^2}\delta_{66}+\frac{\sqrt{n_x}}{1-\epsilon^2}\frac{\sqrt{n_x}}{1-\epsilon_{max}^2}\left(\lVert Q\rVert_F+\lVert A\rVert_F^2\lVert B\rVert_F^2\lVert R\rVert_F \lVert R^{-1}\rVert_F^2\lVert\hat{P}_0\rVert_F^2\right)\delta_{56}$, \\
$\rho_5:=\frac{\sqrt{n_x}}{1-\epsilon^2}\delta_{65}+\frac{\sqrt{n_x}}{1-\epsilon^2}\frac{\sqrt{n_x}}{1-\epsilon_{max}^2}\left(\lVert Q\rVert_F+\lVert A\rVert_F^2\lVert B\rVert_F^2\lVert R\rVert_F \lVert R^{-1}\rVert_F^2 \lVert \hat{P}_0\rVert_F^2\right)\delta_{55}$,\\
$\rho_4:=\frac{\sqrt{n_x}}{1-\epsilon^2}\delta_{64}+\frac{\sqrt{n_x}}{1-\epsilon^2}\frac{\sqrt{n_x}}{1-\epsilon_{max}^2}\left(\lVert Q\rVert_F+\lVert A\rVert_F^2\lVert B\rVert_F^2\lVert R\rVert_F\lVert R^{-1}\rVert_F^2\lVert \hat{P}_0\rVert_F^2\right)\delta_{54}$,\\
$\rho_3:=\frac{\sqrt{n_x}}{1-\epsilon^2}\delta_{63}+\frac{\sqrt{n_x}}{1-\epsilon^2}\frac{\sqrt{n_x}}{1-\epsilon_{max}^2}\left(\lVert Q\rVert_F+\lVert A\rVert_F^2\lVert B\rVert_F^2\lVert R\rVert_F\lVert R^{-1}\rVert_F^2 \lVert\hat{P}_0\rVert_F^2\right)\delta_{53}$,\\
$\rho_2:=\frac{\sqrt{n_x}}{1-\epsilon^2}\delta_{62}+\frac{\sqrt{n_x}}{1-\epsilon^2}\frac{\sqrt{n_x}}{1-\epsilon_{max}^2}\left(\lVert Q\rVert_F+\lVert A\rVert_F^2\lVert B\rVert_F^2\lVert R\rVert_F \lVert R^{-1}\rVert_F^2\lVert\hat{P}_0\rVert_F^2\right)\delta_{52}$,\\
$\rho_1:=\frac{\sqrt{n_x}}{1-\epsilon^2}\delta_{61}+\frac{\sqrt{n_x}}{1-\epsilon^2}\frac{\sqrt{n_x}}{1-\epsilon_{max}^2}\left(\lVert Q\rVert_F+\lVert A\rVert_F^2\lVert B\rVert_F^2\lVert R\rVert_F \lVert R^{-1}\rVert_F^2\lVert\hat{P}_0\rVert_F^2\right)\delta_{51}$,\\
$ \rho_0:=\frac{\sqrt{n_x}}{1-\epsilon^2}\delta_{60}+\frac{\sqrt{n_x}}{1-\epsilon^2}\frac{\sqrt{n_x}}{1-\epsilon_{max}^2}\left(\lVert Q\rVert_F+\lVert A\rVert_F^2\lVert B\rVert_F^2\lVert R\rVert_F\lVert R^{-1}\rVert_F^2\lVert \hat{P}_0\rVert_F^2\right)\delta_{50}.$\\

Now, we can go back to \eqref{ISS111}: for any $i\in \mathbb{Z}_+$
\begin{equation}\label{ISS16}
  \begin{split}
     \lVert \hat{P}_{i+1}- P^*\rVert_F &\leq \lVert \mathcal{L}^{-1}_{\left(A,B,\hat{P}_i\right)}\left(Q+\alpha^\top\left(\hat{P}_i\right)\beta^{-1}\left(\hat{P}_i\right)R\beta^{-1}\left(\hat{P}_i\right)\alpha\left(\hat{P}_i\right)\right)-P^*\rVert_F+ \lVert \varepsilon\left(\Delta A_i,\Delta B_i\right) \rVert_F \\
       &\leq c \lVert \hat{P}_{i}- P^*\rVert_F+\sigma_i \lVert \Delta \theta_i \rVert_F\leq c \lVert \hat{P}_{i}- P^*\rVert_F+\underbrace{\sigma_i \Delta\theta_i^\mathrm{Upper}}_{=:\Omega_i}.
  \end{split}
\end{equation}
The first inequality comes from taking the norm of \eqref{ISS111} and Theorem \ref{theorem5}, and the second inequality comes from \eqref{ISS16}. Then we recursively apply \eqref{ISS16} and get:
\begin{equation}\label{ISS17}
  \begin{split}
     \lVert \hat{P}_{i}- P^*\rVert_F &\leq c \lVert \hat{P}_{i-1}- P^*\rVert_F+\Omega_{i-1}\leq c^2 \lVert \hat{P}_{i-2}- P^*\rVert_F+\left(1+c\right)\lVert \Omega\rVert_{\infty}\\
     &\leq...\leq c^{i}\lVert \hat{P}_{0}- P^*\rVert_F+\left(1+c+c^2+...+c^{i-1}\right)\lVert \Omega \rVert_{\infty} \leq c^{i}\lVert \hat{P}_{0}- P^*\rVert+ \frac{1}{1-c}\lVert \Omega\rVert_{\infty}.
  \end{split}
\end{equation}
Then we finish the proof and get the final result:
\begin{equation}\label{ISS18}
   \lVert \hat{P}_{i}- P^*\rVert_F \leq \beta\left(\lVert \hat{P}_{0}- P^*\rVert_F,i\right)+\gamma(\lVert \Omega\rVert_{\infty}),
\end{equation}
where $\beta\left(\lVert \hat{P}_{0}- P^*\rVert,i\right):=c^i\lVert \hat{P}_{0}- P^*\rVert_F$ is $\mathcal{KL}$ function and $\gamma\left(\lVert \Omega\rVert_{\infty}\right):=\frac{1}{1-c}\lVert \Omega\rVert_{\infty}$ is a $\mathcal{K}$ function. 
\subsection{Proof of Corollary \ref{Coro1}}\label{App3}
From Lemma \ref{theorem21} and Corollary \ref{Coro2}, we know that the sequence $\Delta\theta_i^\mathrm{Upper}$ is non-increasing and converges to $0$ if the sequence is locally persistent with respect to any $N,M,\alpha$. For any $\epsilon>0$, there exists an index $i_1\in \mathbb{Z}_+$, such that $\Delta\theta_{i_1}^\mathrm{Upper}<\frac{\epsilon_1(1-c)}{2\sigma_0}$. Consequently, $\sup\{\Delta\theta_i^\mathrm{Upper}\}_{i=i_1}^{\infty}<\frac{\epsilon_1(1-c)}{2\sigma_0}$. Now, consider $i_2\geq i_1$. For $k\geq i_2$, because $\hat{P}_i$ is bounded, we have the following from \eqref{reformulated}:
\begin{equation}\label{ISS19}
   \begin{split}
      \lVert \hat{P}_{i}- P^*\rVert_F &\leq \beta(\lVert \hat{P}_{i_2}- P^*\rVert_F,i-i_2)+\frac{\epsilon_1}{2}\\
      &\leq \beta(\epsilon_{2},i-i_2)+\frac{\epsilon_1}{2},
   \end{split}
\end{equation}
for some $\epsilon_{2}>0$. Since $\lim\limits_{i\rightarrow\infty}\beta(\epsilon_{2},i-i_2)=0$, there exists an $i_3\geq i_2$ such that $\beta(\epsilon_{2},i-i_2)<\frac{\epsilon_1}{2}$ for all $i\geq i_3$. This completes the proof $\lim\limits_{i\rightarrow\infty}\hat{P}_i=P^*$. Together with Corollary \ref{Coro2}, $\lim\limits_{i\rightarrow\infty}\theta_i=\theta$, we can conclude $\lim\limits_{i\rightarrow\infty}\hat{K}_i=K^*$.
\end{document}